\documentclass[11pt,letterpaper]{article}

\usepackage[margin=1in]{geometry}

\usepackage{hyperref,enumitem}
\setlist{noitemsep,leftmargin=\parindent,topsep=2pt}
\usepackage[numbers]{natbib}
\usepackage{xspace}
\usepackage{graphicx}
\usepackage{amsmath,amssymb,amsthm}
\usepackage{booktabs} %
\usepackage{color}
\usepackage[font=small,labelfont=bf]{caption}
\usepackage{subcaption}
\usepackage{multirow}
\usepackage{enumitem}
\usepackage{authblk}
\usepackage{algorithm}
\usepackage[noend]{algpseudocode}
\usepackage{bbm}
\usepackage{comment}
\usepackage{mathtools}

\usepackage{enumitem}

\newcount\Comments 
\newcommand{\kibitz}[2]{\ifnum\Comments=1{\color{#1}{#2}}\fi}

\newcommand{\es}[1]{\kibitz{red}{}}

\newcommand{\E}{\mathbb{E}}

\DeclareMathOperator*{\argmax}{arg\,max}
\newcommand{\rgtucb}{\textsc{Regret-UCB}}
\newcommand{\init}{\textsc{Init}}
\newcommand{\genbids}{\textsc{Generate-Bids}}
\newcommand{\randombids}{\textsc{Random-Bids}}
\newcommand{\egreedy}{\textsc{$\eps$-Greedy}}
\newcommand{\1}{\mathbbm{1}}

\newcommand{\eps}{\epsilon}

\newcommand{\R}{\mathbb{R}}

\newcommand{\PP}{\mathbb{P}}

\newcommand{\bvec}{\vec{b}}

\theoremstyle{plain}
\newtheorem{theorem}{Theorem}[section]
\newtheorem{corollary}{Corollary}[section]
\newtheorem{lemma}[theorem]{Lemma}
\newtheorem{definition}[theorem]{Definition}
\newtheorem{assumption}[theorem]{Assumption}

\begin{document}
\title{{\bfseries Online Learning for Measuring Incentive Compatibility in Ad Auctions}}

\author[1]{Zhe Feng\footnote{This work was done while the first author was an intern in Core Data Science, Facebook Research.}}
\author[2]{Okke Schrijvers}
\author[2]{Eric Sodomka}
\affil[1]{John A.~Paulson School of Engineering and Applied Sciences, Harvard University\authorcr \texttt{zhe\_feng@g.harvard.edu}}
\affil[2]{Facebook Research\authorcr \texttt{\{okke,sodomka\}@fb.com}}

\date{}
\maketitle

\begin{abstract}
In this paper we investigate the problem of measuring end-to-end Incentive Compatibility (IC) regret given black-box access to an auction mechanism. Our goal is to 1) compute an estimate for IC regret in an auction, 2) provide a measure of certainty around the estimate of IC regret, and 3) minimize the time it takes to arrive at an accurate estimate. We consider two main problems, with different informational assumptions: In the \emph{advertiser problem} the goal is to measure IC regret for some known valuation $v$, while in the more general \emph{demand-side platform (DSP) problem} we wish to determine the worst-case IC regret over all possible valuations. The problems are naturally phrased in an online learning model and we design $\rgtucb$ algorithms for both problems. We give an online learning algorithm where for the advertiser problem the error of determining IC shrinks as $O\Big(\frac{|B|}{T}\cdot\Big(\frac{\ln T}{n} + \sqrt{\frac{\ln T}{n}}\Big)\Big)$ (where $B$ is the finite set of bids, $T$ is the number of time steps, and $n$ is number of auctions per time step), and for the DSP problem it shrinks as $O\Big(\frac{|B|}{T}\cdot\Big( \frac{|B|\ln T}{n} + \sqrt{\frac{|B|\ln T}{n}}\Big)\Big)$. For the DSP problem, we also consider stronger IC regret estimation and extend our $\rgtucb$ algorithm to achieve better IC regret error. We validate the theoretical results using simulations with Generalized Second Price (GSP) auctions, which are known to not be incentive compatible and thus have strictly positive IC regret.
\end{abstract}

\section{Introduction}\label{sec:intro}
Online advertising has grown into a massive industry, both in terms of ads volume as well as the complexity of the ecosystem. Most inventory is sold through auctions, and demand-side platforms (DSPs)---such as AppNexus, Quantcast, Google's DoubleClick and Facebook's Audience Network---offer advertisers the opportunity buy inventory from many different publishers. In turn, publishers themselves may run their own auctions to determine which ad to show and what price to charge, resulting in a complicated sequence of auction systems. While auction theory has influenced the design of online ad auctions, in practice not all auctions are incentive compatible. Practices such as sequential auctions (where the winners of an auction on a DSP compete in further auctions on the publisher website) and ``dynamic reserve prices'' (where publishers use a reserve price of e.g. $50\%$ of the highest bid in the auction) both violate incentive compatibility, even when individual auctions are otherwise designed to be incentive compatible.

A lack of incentive compatibility in the ecosystem is problematic both for advertisers and DSPs. Advertisers have to spend time and energy in figuring out the optimal bidding strategy for different placements, a process which can be quite costly especially when incentives differ across different publishers. For DSPs this is problematic because common budget management techniques, such as multiplicative pacing, are only optimal when the individual auctions are incentive compatible \cite{CKSS18}.

Since (the lack of) incentive compatibility impacts bidding strategies for advertisers and the quality of the product that a DSP offers, it's important to understand the end-to-end incentive compatibility of an ad system. But how should we think of quantifying incentive compatibility? IC is a binary property: either a buyer maximizes their utility by bidding truthfully or she doesn't. However, this fails to capture that a first-price auction is ``worse,'' in some sense, than a second-price auction with a dynamic reserve price of $50 \%$ of the highest bid. To capture these differences, we focus on dynamically measuring \emph{IC regret}:
\begin{align}
\text{IC regret}(v_i) = \max_{b_i} \E_{b_{-i}}\left[ u_i(b_i, b_{-i}) - u_i(v_i, b_{-i})\right],\label{eq:ic-regret}
\end{align}
where $v_i$ is the true value of advertiser $i$, $b_i$ the bid of $i$, $b_{-i}$ the bids of other advertisers, and $u_i(\cdot)$ the (expected) utility of $i$. IC regret captures the difference in utility between bidding truthfully, and the maximum utility achievable. By definition, incentive compatible mechanisms have IC regret $0$, while higher IC regret indicates a stronger incentive to misreport.

In this paper we focus on measuring IC regret assuming only black-box access to the auction mechanism. That is, for a given advertiser bid $b_i$ we observe whether this yielded an impression $x_i$, and what price $p_i$ was paid if any. While it seems restrictive to only focus on black box access, this setting allows one to deal with the fact that auction logic is spread out over different participants in the ecosystem (e.g. a DSP and publisher may each run an auction using proprietary code), and it even allows an advertiser who has no access to any auction code to verify the end-to-end incentives in the system.

Given black-box access to the auction mechanism, our goal is to 1) compute an estimate for IC regret in an auction, 2) provide a measure of certainty around the estimate of IC regret, and 3) minimize the time it takes to arrive at an accurate estimate.

We approach this problem using tools from the combinatorial (semi-)bandit literature: to measure the IC regret in an auction system, a bidder may enter different bids in different auctions\footnote{In this paper, we consider the model that a bidder participates in many auctions and partitions them into several blocks. For every auction in each block, the bidder enters a bid and only observes the allocation and payment of the bid he enters from each block. For more formal definition, see Section~\ref{sec:online-learning-model}.}. By judiciously choosing these bids, we give bounds on the difference between the measured IC regret, and the true IC regret of the system as a function of the number of auctions the bidder participated in. This provides a trade-off between the accuracy of the IC regret measurement, and the the amount of time (and thus money) spend on obtaining the insight.

We consider the problem from two perspectives. The first, which we call the ``Advertiser Problem,'' consists of determining the IC regret in the system for some \emph{known} true value $v$. This is the problem the advertiser faces when they are participating in an online ad auction. In Section~\ref{sec:fix-value} we give an algorithm to select bids that determine IC regret with an error that shrinks as $O\Big(\frac{|B|}{T}\cdot\Big(\frac{\ln T}{n} + \sqrt{\frac{\ln T}{n}}\Big)\Big)$, where $B$ is the (finite) bid space, $n$ is the number of auctions at each time step and $T$ is the total number of time steps. The second problem, which we call the ``DSP Problem'' consists of determining IC regret for the worst-case valuation $v$, i.e. the value $v$ that yields the \emph{highest} IC regret. This is the problem that a DSP faces if they want to ensure that IC regret is low for \emph{all} advertisers that may use their services. In Section~\ref{sec:non-fix-value} we give an algorithm to select bids and (hypothetical) values that determine worst-case IC regret over all possible values where the error term shrinks as $O\Big(\frac{|B|}{T}\cdot\Big( \frac{|B|\ln T}{n} + \sqrt{\frac{|B|\ln T}{n}} \Big)\Big)$ as the number of time steps $T$ grows.

The rest of the paper is organized as follows. Section 2 illustrates the main model and preliminaries in this paper. Sections 3 and Section 4 discuss two main problems about measuring IC regret in practice, the advertiser problem and the DSP problem. In Section 5, we consider the stronger IC regret estimation for the DSP problem and extend our $\rgtucb$ algorithm to achieve better IC regret error. Moreover, we discuss the continuous bids space setting. In Section 6, we run some simulations with GSP auctions to validate our theoretical results. Finally, we conclude our paper.
\vspace{-5pt}
\subsection{Related Work}
Structurally, designing incentive compatible (IC) mechanisms is well-understood: if the corresponding optimization problem can be solved exactly then the VCG mechanism \cite{V61,C71,G73} is IC. In cases where computing the optimal allocation is prohibitively expensive, and some approximation algorithm is used, VCG is no longer incentive compatible \cite{NR07}. At the same time, the structure of IC mechanisms beyond optimal allocation is known, and tied to some monotonicity property of the allocation function for single-parameter settings \cite{M81} and more generally \cite{R87}. A weakness in these approaches is that determining if some given mechanism $M$ is IC, one needs to know both the allocation \emph{function} and the pricing \emph{function}. Without access to the code that computes the allocation and pricing, the best an advertiser (or a DSP who depends on a third party auction that's run by a publisher) can hope to have access to is samples from the allocation and pricing functions, not the functions themselves.

\citet{LMSV18} are the first to study the problem of testing incentive compatibility assuming only black box access. They design an A/B experiment to determine whether an auction is incentive compatible in both single-shot and dynamic settings (the latter concerns cases where an advertiser's bid is used to set their reserve price in later auctions). The work of \citet{LMSV18} provides a valuable tool for validation of incentive compatibility, but leaves open the question of how to design an experiment that minimizes the time required to get results with high confidence. The present paper complements \cite{LMSV18} by giving an algorithm to select alternative bids that minimize the error in the estimate for IC regret.

Our work lies at the intersections of two areas, {\em No regret learning in Game Theory and Mechanism Design} \cite{CHN14} and {\em Contextual Bandits} \cite{BCB2012, Agarwal14}. On the one hand, in Game Theory and Mechanism Design, most of the existing literature focuses on maximizing revenue for the auctioneer without knowing a priori the valuations of the bidders \cite{BKRW04, ARS14, ACDKR15}, as well as optimizing the bidding strategy from the bidder side \cite{WRP16, BG17, FPS18}. To the best of our knowledge, this is the first paper on testing IC through online learning approach. On the other hand, our work is related to the literature on {\em Contextual Bandits}. In the terminology of contextual bandits, the discrete bids in our work correspond to arms, and the observed average allocation and payment of the selected bids at each time correspond to different contexts.

\subsection{Contributions}
There are three main contributions of this paper.

\begin{enumerate}
\item We build an online learning model to measure IC regret with only black-box access to the auction mechanism. We present algorithms that find IC regret with an error that shrinks as $O\Big(\frac{|B|}{T}\cdot\Big(\frac{\ln T}{n} + \sqrt{\frac{\ln T}{n}}\Big)\Big)$ for a known valuation $v$, and an error  that shrinks as $O\Big(\frac{|B|}{T}\cdot\Big( \frac{|B|\ln T}{n} + \sqrt{\frac{|B|\ln T}{n}} \Big)\Big)$ for all possible valuations $v$. In the above,  $B$ is the (finite) bid space, $n$ is the number of auctions at each time step and $T$ is the total number of time steps. We also extend our $\rgtucb$ algorithm to handle stronger IC regret estimation and shrink the error as $O\left(\frac{|B|^2(\ln T)^{2/3}}{n^{2/3}T}\right)$.

\item We present a combinatorial semi-bandit algorithm for the setting where the "observed reward" is the $\mathrm{max}$ over several arms selected at each time while the benchmark is the standard optimal arm with highest expected reward. To understand this mismatch between the benchmark and observed reward, we analyze the {\em Pseudo-Regret} of the algorithm and uncover its trade-off with the number of arms selected at each time. This analysis may be of independent interest to the combinatorial (Semi-)Bandits literature.
\item Simulations suggest that, (1) there is a trade-off between {\em Pseudo-Regret} and the number of blocks to partition auctions, i.e. we would like to choose number of blocks neither very small nor very large, (2) the {\em Pseudo-Regret} (near linearly) decays when the number of auctions accessed at each time grows, and (3) our designed algorithm performs better than naive $\randombids$ algorithm and $\egreedy$ algorithm.
\end{enumerate}
\vspace{-5pt}
\section{Model and Preliminaries}\label{sec:prelim}
In this section, we formally define the model we considered in this paper and present some important preliminaries. For simplifying description, we introduce "test bidder" to model both the advertiser and DSP.
\vspace{-10pt}
\subsection{Auctions for a Test Bidder}
Consider a test bidder $i$ who is eligible for $n$ auctions every day. In each auction, they submit a bid $b_i$, the auction is run and the outcome is determined by the allocation rule $g_i: b_i \rightarrow[0, 1]$ and the payment rule $p_i: b\rightarrow\R_{\geq 0}$ (where $g_i$ and $p_i$ are conditioned on the competition in the auction). We assume a single-parameter quasilinear utility model $u_i(v_i, b_i) = g_i(b_i)\cdot v_i - p_i(b_i)$ for some true value $v_i\in \R_{\geq 0}$. Let $U$ be a bound on the utility of a buyer, i.e. the utility function is bounded by $[-U, U]$.

Since there is a large number of auctions and other bidders in the ecosystem, we  model the randomness in the system using a stationarity assumption where the conditional allocation and pricing functions are drawn from the same underlying distribution for each auction $k$:
\begin{assumption} [Stochastic Assumption]
For each auction $k$, the allocation rule $g_i$ and payment rule $p_i$ of the test bidder $i$ are drawn from an unknown distribution $\mathcal{F}$, i.e. $(g_i, p_i)_k \overset{\text{\footnotesize i.i.d.}}{\sim} \mathcal{F}$.
\end{assumption}
\vspace{-10pt}
\subsection{Online Learning Model}\label{sec:online-learning-model}
At every time step $t$, the test bidder participates $n$ auctions and randomly partitions them into $m+1$ blocks of equal size.\footnote{The number of blocks $m$ is a variable that will be fixed later. The choice of which $m$ to use will trade off the number of different bids we get information on and the accuracy of the outputs for each bid. Throughout the paper, we assume $m < |B|$.} For $j\in [m+1]$, let $\mathcal{A}_j$ be the set of auctions in block $j$.\footnote{We let $[n] = \{1, 2, 3, ..., n\}$ be the set of positive integers up to and including $n$.} For every auction in each block $j$, the test bidder submits a bid $b_t^j\in B$, where $B$ is a finite set of bids.\footnote{While in theory the bid space is continuous, in practice outcomes for different bids are relatively smooth, see e.g. the plots in \cite{LMSV18}, so discretizing the bid space is a reasonable simplification. For more discussion, see Section~\ref{sec:continuous-bid-space}.} %

At the end of time $t$, in each block $j\in [m]$, the test bidder observes the average allocation probability $\widetilde{g}_t^j(b_t^j)$ and average payment $\widetilde{p}_t^j(b_t^j)$ over all auctions in block $j$. Let $\bvec_t = \{b_t^1, b_t^2,\cdots,b_t^m\}$ be the bids of the test bidder at time $t$. For block $m+1$, the test bidder bids their true value $v$ and observes the average allocation $\widetilde{g}^{m+1}_t(v)$ and payment $\widetilde{p}^{m+1}_t(v)$. Without loss of generality, we also assume $v\in B$. The average utility of the test bidder in each block $j\in[m+1]$ is $\widetilde{u}^j_t (v, b_t^j)\equiv \widetilde{g}^j_t(b_t^j)\cdot v - \widetilde{p}^j_t(b_t^j)$.

With the setting defined, we can instantiate an empirical version of the IC regret from Equation~\eqref{eq:ic-regret} for a given true value $v$ and bid $b_t^j$:

\begin{definition}[Empirical IC regret]\label{def:emp-regret}
\begin{equation}
\widetilde{rgt}_t\left(v, \bvec_t\right) = \max_{j\in[m]}\left\{\widetilde{u}^j_t (v, b^j_t) - \widetilde{u}^{m+1}_t (v, v)\right\}.
\end{equation}
\end{definition}

For notational simplicity, we denote $g^*(\cdot) \triangleq \E\left[\widetilde{g}^j_t(\cdot)\right]$ and $p^*(\cdot) \triangleq \E\left[\widetilde{p}^j_t(\cdot)\right]$ for each block $j$ at time $t$.\footnote{Throughout this paper, the expectation $\E$ is over all the randomness (e.g., randomization over payment and allocation distribution and randomization over algorithm).} Similarly, define $u^*(v, b)  \triangleq \E\left[\widetilde{u}^j_t(v, b)\right] \equiv g^*(b)\cdot v - p^*(b)$. Given these definitions, {\em IC regret} for a particular true value $v$ and bid $b$ is:
\begin{definition}[IC-regret]\label{def:regret}
$rgt(v, b)=u^*(v, b) - u^*(v, v)$.
\end{definition}

Thus far, in Definitions~\ref{def:emp-regret}~and~\ref{def:regret} we've considered the regret that a buyer has for bidding their true value $v$ compared to a particular alternative bid $b$. The quantity that we're really interested in (cf. Equation~\eqref{eq:ic-regret}) is the worst-case IC regret with respect to different possible bids. For the Advertiser Problem (which we treat in Section~\ref{sec:fix-value}) this is with respect to some known true value $v$, whereas for the DSP Problem (which we treat in Section~\ref{sec:non-fix-value}) we consider the worst-case IC regret over all possible true values $v$. To summarize, the learning task of the test bidder is to design an efficient learning algorithm to generate $v_t$ and $b_t$ in order to minimize the following {\em Pseudo-Regret}, the difference between cumulative empirical regret and benchmark.\footnote{Where in the case of the Advertiser Problem, we define Pseudo-Regret with respect to a given value $v$ analogously.}

\begin{definition}[Pseudo-Regret~\cite{BCB2012}]
\begin{equation}\label{eq:pseudo-regret}
\E[R(T)] = \max_{v, b\in B}\sum_{t=1}^T rgt(v, b) - \E\left[\sum_{t=1}^T\widetilde{rgt}_t(v_t, \bvec_t)\right]
\end{equation}
\end{definition}

\es{Add a sentence to explain?}

Given the above {\em Pseudo-Regret} definition, we can define the error of determining IC regret, which is to measure the distance between the optimal IC regret and the average empirical IC regret over time.
\begin{definition}[IC Regret Error] $\mathcal{E}(T) = \frac{\E[R(T)]}{T}$
\end{definition}
\subsection{More Related work}
The online learning model in our work is also related to Combinatorial (Semi)-bandits literature \cite{Gai2012, CWY13, KWAS15}. \citet{CWY13} first proposed the general framework and Combinatorial Multi-Armed Bandit problem, and our model lies in this general framework (i.e. the bids are generated from a super %
arm at each time step). However, in our model the test bidder can observe feedback from multiple blocks at each time step, which is similar to the combinatorial semi-bandits feedback model. \citet{Gai2012} analyze the combinatorial semi-bandits problem with linear rewards, and \citet{KWAS15} provides the tight regret bound for the same setting. The $\rgtucb$ algorithm in our paper is similar and inspired by the algorithm in \cite{Gai2012, KWAS15}. However, in our model, the reward function is not linear and has a ``$\max$'', which needs more work to address it. \citet{CHLLLL16} firstly consider the general reward function in Combinatorial Bandits problem and can handle ``$\max$'' in reward function, however, there is a mismatch between the benchmark and observed reward in our model (see more in "Contributions").

\section{The Advertiser Problem}\label{sec:fix-value}
In this section, we focus on the special case of measuring the IC regret for a known true value $v$. The learning problem in this setting is to select, for each timestep $t$, a bid profile $\bvec_t$ of $m$ bids to be used in the $m$ auction blocks at that time. The bids $\bvec_t$ should be selected to minimize the {\em Pseudo-Regret} defined in Equation~\ref{eq:pseudo-regret}.

We propose the $\rgtucb$ algorithm, given in Algorithm~\ref{alg:rgtucb-multiple-bids}, which is inspired by the $\mathrm{CombUCB1}$ algorithm \cite{Gai2012, KWAS15}. $\rgtucb$ works as follows: at each timestep $t$, the algorithm performs three operations. First, it computes the {\em upper confidence bound} (UCB) on the expected utility of each bid $b\in B$ given $v$,
\begin{equation}\label{eq:ucb-utility}
\mathrm{UCB}^\text{u}_t(v, b) = \widehat{g}_{t-1}(b)\cdot v - \widehat{p}_{t-1}(b) + 2U\sqrt{\frac{2(m+1)\ln t}{n_{t-1}(b)\cdot n}},
\end{equation}
where $\widehat{g}_s(b)$ is the average allocation probability of bid $b$ up to time $s$, $n_s(b)$ is the number of times that bid $b$ has been submitted in $s$ steps and $U$ is the bound of utility (i.e. utility is bounded by $[-U, U]$). Second, the algorithm chooses the $m$ bids that correspond to the largest $\mathrm{UCB}^\text{u}_t$ values; call this bid vector $\bvec_t$. The algorithm uses the $m$ bids $\bvec_t$ in the first $m$ blocks of auctions, and in the final block of auctions, it uses the bidder's true value $v$. Finally, for each of the blocks, the algorithm observes the average allocation probability and payment, and it uses that to update the estimates of allocation $\widehat{g}$ and payment $\widehat{p}$.

\begin{algorithm}[H]
	\begin{algorithmic}
		\State \textbf{Input:} A finite set of bids $B$, parameter $m$ and $n$
		\For{$b$ in $B$}
		\State Randomly participate in $\frac{n}{m+1}$ auctions and submit bid $b$ for each auction.
		\State Observe the average allocation $\widetilde{g}(b)$ and payment $\widetilde{p}(b)$.
		\State $\widehat{g}_0 (b)\leftarrow\widetilde{g}(b)$, $\widehat{p}_0 (b)\leftarrow\widetilde{p}(b)$.
		\EndFor
	\end{algorithmic}
	\caption{$\init$ algorithm to get $\widehat{g}_0$ and $\widehat{p}_0$.\label{alg:init}}
\end{algorithm}

Given a fixed valuation $v$, define $b^*$ to be the best-response bid: $$b^* \triangleq \argmax_{b\in B} rgt(v, b) = \argmax_{b\in B}u^*(v, b).$$
and denote $\Delta(b) \triangleq u^*(v, b^*) - u^*(v, b)$. The following theorem bounds the worst case {\em Pseudo-Regret} of the $\rgtucb$ algorithm for known valuation $v$.
\begin{theorem}\label{thm:rgt-bound-multiple-bids}
$\rgtucb$ achieves pseudo-regret at most $$\sum_{b\in B: u^*(v, b) < u^*(v, b^*)} \frac{32(m+1)U^2\ln T}{n\Delta(b)} + \frac{\pi^2}{3}\cdot\frac{\Delta(b)}{m}$$
\end{theorem}

\begin{algorithm}[H]
\begin{algorithmic}
\State \textbf{Input:} A finite set of sbids $B$, parameter $m, n$. $\forall b\in B, n_0(b)=1$
\State \textbf{Initialize:} Run $\init(B, m,n)$ algorithm to get $\widehat{g}_0$ and $\widehat{p}_0$. 
\For{$t = 1,\cdots, T$}
\State Update $\mathrm{UCB}$ terms on the expected utility of each bid $b$
\State $\forall b\in B, \mathrm{UCB}^\text{u}_t(v, b)\leftarrow \widehat{g}_{t-1}(b)\cdot v - \widehat{p}_{t-1}(b) + 2U\sqrt{\frac{2(m+1)\ln t}{n_{t-1}(b)\cdot n}}$
\State Generate a sequence of different $m$ bids $\bvec_t \in B^m$ to maximize
\begin{equation}
\sum_{b\in \bvec_t} \mathrm{UCB}^\text{u}_t(v, b)
\end{equation}
\For{$j = 1,\cdots,m$}
\State Update $n_t(b_t^j)\leftarrow n_{t-1}(b_t^j) + 1$
\State Observe $\tilde{g}_t^j(b_t^j)$ and $\tilde{p}^j_t(b_t^j)$
\State $$\widehat{g}_t(b_t^j) \leftarrow \left[\widehat{g}_{t-1}(b_t^j) \cdot n_{t-1}(b_t^j) + \tilde{g}^j_t(b_t^j)\right]\big/ n_t(b_t^j)$$
\State $$\widehat{p}_t(b_t^j) \leftarrow \left[\widehat{p}_{t-1}(b_t^j) \cdot n_{t-1}(b_t^j) + \tilde{p}^j_t(b_t^j)\right]\big/n_t(b_t^j)$$
\EndFor
\For{$b\notin \bvec_t$}
\State $n_t(b) \leftarrow n_{t-1}(b)$, $\widehat{g}_t(b)\leftarrow \widehat{g}_{t-1}(b)$, $\widehat{p}_t(b)\leftarrow \widehat{p}_{t-1}(b)$
\EndFor
\EndFor
\end{algorithmic}
\caption{$\rgtucb$ Algorithm for a known valuation.}\label{alg:rgtucb-multiple-bids}
\end{algorithm}

To prove the theorem, we rely on the following lemma%
, which is widely used in stochastic bandit literature. We refer the reader to check \cite{BCB2012} for more material.%
\begin{lemma}[\cite{BCB2012}]\label{lem:ucb-utility}
Fix a valuation $v$ and an iteration $t-1$ (where $t\geq 2$). If a bid $b\neq b^*$ (i.e. $u^*(v, b) < u^*(v, b^*)$) has been observed $n_{t-1}(b)\geq \frac{8(m+1)(2U)^2\ln t}{n\Delta(b)^2}$ times, then with probability of at least $1 - \frac{2}{t^2}$, $\mathrm{UCB}^\text{u}_t(v, b)\leq \mathrm{UCB}^\text{u}_t(v, b^*)$.
\end{lemma}

\subsection{Proof of Theorem~\ref{thm:rgt-bound-multiple-bids}}\label{sec:proof-rgt-bound1}
First we need some additional notations. Let $\zeta_t = \Big\{\forall b\in \bvec_t, n_{t-1}(b)\geq\frac{8(m+1)(2U)^2\ln t}{n\Delta(b)^2}\Big\}$ be the event that every bid in $\bvec_t$ has previously been used at least $\frac{8(m+1)(2U)^2\ln t}{n\Delta(b)^2}$ times at time $t$, $\Delta(\bvec_t) \triangleq u^*(v, b^*) - \max_{b \in \bvec_t} u^*(v, b)$\footnote{For notation simplicity, we still use $"\Delta"$ here to represent the minimum gap of the expected utility between the bids in $\bvec_t$ and the optimal bid $b^*$} \es{Change $\Delta$ to something else?}, and event $\zeta_{t, b} = \Big\{n_{t-1}(b)\geq\frac{8(m+1)(2U)^2\ln t}{n\Delta(b)^2}\Big\}$. The proof consists of three parts. First we upper bound the {\em Pseudo-Regret} in the following way,
\begin{equation}\label{eq:ub-pseudo-regret1}
\begin{aligned}
&\E[R(T)] = \max_{b\in B}\left[\sum_{t=1}^T u^*(v, b) - u^*(v, v)\right] - \E\left[\sum_{t=1}^T\widetilde{rgt}_t(v, \bvec_t)\right]\\
& = \max_{b\in B}\left[\sum_{t=1}^T g^*(b)\cdot v - p^*(b) - \Big(g^*(v)\cdot v - p^*(v)\Big)\right] \\
& \hspace{.5cm} - \E\Bigg[\sum_{t=1}^T\max_{j\in [m]}\Big\{\widetilde{g}^j_t(b_t^j)\cdot v - \widetilde{p}^j_t(b_t^j)\Big\} - \Big(\widetilde{g}^{m+1}_t(v)\cdot v - \widetilde{p}^{m+1}_t(v)\Big)\Bigg]\\
& \overset{(a)}{=} \max_{b\in B}\sum_{t=1}^T g^*(b)\cdot v - p^*(b) - \E\left[\sum_{t=1}^T \max_{j\in [m]}\Big\{\widetilde{g}^j_t(b_t^j)\cdot v - \widetilde{p}^j_t(b_t^j)\Big\}\right]\\
& \overset{(b)}{=} \sum_{t=1}^T g^*(b^*)\cdot v - p^*(b^*) - \sum_{t=1}^T \E\left[\E\Big[\max_{j\in [m]} \widetilde{g}^j_t(b_t^j)\cdot v - \widetilde{p}^j_t(b_t^j)\Big]\Big| \bvec_t\right]\\
& \overset{(c)}{\leq} \sum_{t=1}^T g^*(b^*)\cdot v - p^*(b^*) - \sum_{t=1}^T \E\left[\max_{j\in [m]} \E\Big[\widetilde{g}^j_t(b_t^j)\cdot v - \widetilde{p}^j_t(b_t^j)\Big| \bvec_t\Big]\right]\\
& = \sum_{t=1}^T u^*(b^*) - \E\left[\max_{j\in [m]} u^*(b_t^j)\right] \leq \E\left[\sum_{t=1}^T\1\big\{b^*\notin \bvec_t\big\}\cdot\Delta(\bvec_t)\right]
\end{aligned}
\end{equation}

Equality $(a)$ follows from the facts that $\E\left[\widetilde{g}^{m+1}_t(v)\right] = g^*$ and $\E\left[\widetilde{p}^{m+1}_t(v)\right] = p^*$.
For equality $(b)$, the outer expectation is over the history of the test bidder up to iteration $t-1$, which determines $b_t$ and the inner expectation is the expected maximum utility given $b_t$ (over the randomness of the auctions themselves). Inequality $(c)$ holds because $\E[\max\{X_1,\cdots,X_n\}] \geq \max\{\E[X_1],\cdots,\E[X_n]\}$.

Second, we decompose the upper bound of $\E[R(T)]$ in~(\ref{eq:ub-pseudo-regret1}) into two terms. Then we have
\begin{align*}
&\E[R(T)] \leq \E\left[\sum_{t=1}^T\1\big\{b^*\notin \bvec_t\big\}\cdot \Delta(\bvec_t)\right] = \underbrace{\E\left[\sum_{t=1}^T\1\big\{b^*\notin \bvec_t, \zeta_t\big\}\cdot \Delta(\bvec_t)\right]}_{\textbf{Term 1}} + \underbrace{\E\left[\sum_{t=1}^T\1\big\{b^*\notin \bvec_t, \neg \zeta_t\big\}\cdot \Delta(\bvec_t)\right]}_{\textbf{Term 2}}
\end{align*}
Then, we bound \textbf{Term 1} and \textbf{Term 2} respectively. First, we bound \textbf{Term 1} in Equations~(\ref{ie:fix-value-term1}). The first inequality is based on the fact that for any random variables $A_1,\cdot, A_n$ and $B$, 
\begin{eqnarray*}
\1\big\{\forall i, A_i\geq B\big\}\leq \sum_{i=1}^n \1\big\{A_i\geq B\big\}
\end{eqnarray*}
The second inequality is based on the exchange of two summands. The third inequality holds because
\begin{align*}
\E\left[ \1\big\{\mathrm{UCB}^\text{u}_t(v,b)\geq \mathrm{UCB}^\text{u}_t(b^*), \zeta_t\big\}\right]%
\leq \PP\big(\mathrm{UCB}^\text{u}_t(v,b)\geq \mathrm{UCB}^\text{u}_t(b^*)\big|\zeta_t\big),
\end{align*}
and the forth inequality is based on Lemma~\ref{lem:ucb-utility}.
\begin{align}
&\E\left[\sum_{t=1}^T\1\big\{b^*\notin \bvec_t, \zeta_t\big\}\cdot \Delta(\bvec_t)\right]\nonumber
\leq \E\left[\sum_{t=1}^T\1\big\{\forall b\in \bvec_t, \mathrm{UCB}^\text{u}_t(v,b)\geq \mathrm{UCB}^\text{u}_t(v,b^*), \zeta_t\big\}\cdot \Delta(\bvec_t)\right]\nonumber\\
& \leq\E\left[\sum_{t=1}^T\frac{1}{m}\sum_{b\in \bvec_t} \1\big\{\mathrm{UCB}^\text{u}_t(v,b)\geq \mathrm{UCB}^\text{u}_t(v,b^*), \zeta_t\big\}\cdot \Delta(b)\right]\nonumber\\
& \leq\frac{1}{m}\quad \smashoperator{\sum_{\substack{b\in B,\\ u^*(v,b) < u^*(v,b^*)}}}\quad \Delta(b)\cdot \E\left[\sum_{t=1}^T \1\big\{\mathrm{UCB}^\text{u}_t(v,b)\geq \mathrm{UCB}^\text{u}_t(v,b^*), \zeta_{t,b}\big\}\right]\nonumber\\
& \leq \frac{1}{m}\quad \smashoperator{\sum_{\substack{b\in B,\\ u^*(v,b) < u^*(v,b^*)}}}\quad \Delta(b)\cdot\sum_{t=1}^T \PP\left(\mathrm{UCB}^\text{u}_t(v,b)\geq \mathrm{UCB}^\text{u}_t(v, b^*)\Big|\zeta_{t,b}\right)\nonumber\\
& \leq\frac{1}{m}\sum_{\substack{b\in B,\\ u^*(v,b) < u^*(v,b^*)}}\Delta(b)\cdot\Big(1 + \sum_{t=2}^T \frac{2}{t^2}\Big)\nonumber\\
& \leq \frac{1}{m}\sum_{\substack{b\in B,\\ u^*(v,b) < u^*(v,b^*)}}\Delta(b)\cdot\frac{\pi^2}{3} \label{ie:fix-value-term1}
\end{align}
What remains is to bound \textbf{Term 2}. In the following, let $H= \frac{8(m+1)(2U)^2\ln t}{n\Delta(b)^2}$, and we bound \textbf{Term 2} as follows,

\begin{align}
&\E\left[\sum_{t=1}^T\1\big\{b^*\notin \bvec_t, \neg \zeta_t\big\}\cdot \Delta(\bvec_t)\right]\nonumber \leq \E\left[\sum_{t=1}^T \1\big\{b^*\notin \bvec_t, \exists b\in \bvec_t: n_{t-1}(b)\leq H\big\}\cdot \Delta(\bvec_t)\right] \nonumber\\
& \leq  \E\left[\sum_{t=1}^T \sum_{b\in \bvec_t} \1\big\{b\neq b^*, n_{t-1}(b)\leq H\big\}\cdot \Delta(b)\right]\nonumber \\
& \leq \sum_{\substack{b\in B,\\ u^*(v,b) < u^*(v,b^*)}} \Delta(b) \cdot \E\left[\sum_{t=1}^T \1\big\{n_{t-1}(b)\leq H\big\}\right]\nonumber\\
& \leq \sum_{\substack{b\in B,\\ u^*(v,b) < u^*(v,b^*)}} \frac{8(m+1)(2U)^2\ln T}{n\Delta(b)}\label{ie:fix-value-term2}
\end{align}
Combining the inequalities~(\ref{ie:fix-value-term1}) and (\ref{ie:fix-value-term2}), we complete the proof.

\subsection{Discussion}
The $\rgtucb$ algorithm can also be used to implement a low-regret bidding agent: Consider an advertiser who knows the valuation $v$ and wants to maximize the expected utility $u^*(v, b)$ by seeking for a best response bid $b$. Indeed, the advertiser can adopt the exact same algorithm -- $\rgtucb$ to maximize the utility. The analysis for regret bound in Section~\ref{sec:proof-rgt-bound1} also works when we change the reward function from IC regret to utility. 

The term $m$ appears in both terms in Theorem~\ref{thm:rgt-bound-multiple-bids}, hence we can pick it to minimize the asymptotic {\em Pseudo-Regret}:
\begin{corollary}\label{coro:fix-value-error}
Let $\overline{\Delta} \triangleq \max_{b\in B} \Delta(b)$ and $\underline{\Delta} \triangleq \min_{b\in B} \Delta(b)$. Choosing $m = \frac{\pi}{4U}\sqrt{\frac{n\overline{\Delta} \underline{\Delta}}{6\ln T}}$, the error of determining IC regret achieved by $\rgtucb$ algorithm is upper bounded by
\begin{equation*}
\mathcal{E}(T)\leq \frac{|B|}{T} \cdot \left(8\pi U\cdot\sqrt{\frac{\overline{\Delta}\ln T}{6n\underline{\Delta}}} + \frac{32U^2\ln T}{n \underline{\Delta}}\right) = O\left(\frac{|B|}{T}\cdot\Bigg(\frac{\ln T}{n} + \sqrt{\frac{\ln T}{n}}\Bigg)\right)
\end{equation*}
\end{corollary}

\section{The DSP Problem}\label{sec:non-fix-value}
In this section, we consider the problem of determining the worst-case IC regret over all possible valuations $v$. Let $v^*$ and $b^*$ be the value and bid combination that yields the highest {\em IC regret},\footnote{We assume, without loss of generality, that this optimal combination is unique throughout the paper.} i.e. $(v^*, b^*) = \argmax_{v, b \in B} rgt(v, b)$, We modify our previous $\rgtucb$ algorithm in this setting and show the pseudocode in Algorithm~\ref{alg:rgtucb-multiple-bids-nonfixed-value}.

At each time $t$, the algorithm first computes the $\mathrm{UCB}$s on the expected {\em IC regret} of each valuation and bid pair $(v, b)$.
\begin{equation}\label{eq:ucb-rgt}
\mathrm{UCB}^\text{rgt}_t(v, b) = \widehat{rgt}_t(v, b) +  4U\sqrt{\frac{3(m+1)\ln t}{n\cdot \big(n_{t-1}(v)\wedge n_{t-1}(b)\big)}}
\end{equation}
where $\widehat{rgt}_t(v, b) = \widehat{g}_t(b)\cdot v -\widehat{p}_t(b) - \big(\widehat{g}_t(v)\cdot v -\widehat{p}_t(v)\big)$, "$\wedge$" is the $\min$ function, 
and the other notation $\widehat{g}_t$, $\widehat{p}_t$, $n_{t-1}(b)$, $n_{t-1}(v)$, and $U$ are identical to Section~\ref{sec:fix-value}. Then the algorithm selects $(v_t, b_t^1)$ to maximize the $\mathrm{UCB}^\text{rgt}_t$ term. Given the valuation $v_t$, the algorithm chooses other $m-1$ bids to achieve the $(m-1)-$largest $\mathrm{UCB}^\text{u}_t(v_t, \cdot)$ terms defined in Equation~\ref{eq:ucb-utility}. The rest of update steps in the algorithm are exactly the same as in Algorithm~\ref{alg:rgtucb-multiple-bids}.

\begin{algorithm}
\begin{algorithmic}
\State \textbf{Input:} A finite set of bids $B$, parameter $m, n$. $\forall b\in B, n_0(b)=1$
\State \textbf{Initialize:} Run $\init(B, m,n)$ algorithm to get $\widehat{g}_0$ and $\widehat{p}_0$. 
\For{$t = 1,\cdots, T$}
\State Update $\mathrm{UCB}^\text{rgt}$ terms of every $(v, b)$ pair
\begin{equation*}
\mathrm{UCB}^\text{rgt}_t(v, b) = \widehat{rgt}_t(v, b) + 4U\sqrt{\frac{3(m+1)\ln t}{n\cdot \big(n_{t-1}(v)\wedge n_{t-1}(b)\big)}}
\end{equation*}
\State Choose $(v_t, b_t^1)\in B\times B$ to maximize $\mathrm{UCB}^\text{rgt}_t(v, b)$
\If{$m \geq 2$}
\State Choose remaining $m-1$ bids $\{b_t^2,\cdots,b_t^m\}$\footnotemark to maximize $$\sum_{b\in \bvec_t\backslash b_t^1}\mathrm{UCB}^\text{rgt}(v_t, b)$$ %
\EndIf
\State Update $n_t(v_t) \leftarrow n_{t-1}(v_t) + 1$
\State Observe $\widetilde{g}^{m+1}_t(v_t)$ and update
$$\widehat{g}_t(v_t) \leftarrow \left[\widehat{g}_{t-1}(v_t)\cdot n_{t-1}(v_t) + \widetilde{g}^{m+1}_t(v_t)\right]\big/n_{t}(v_t)$$
$$\widehat{p}_t(v_t) \leftarrow \left[\widehat{p}_{t-1}(v_t)\cdot n_{t-1}(v_t) + \widetilde{p}^{m+1}_t(v_t)\right]\big/n_t(v_t)$$
\For{$j = 1,\cdots,m$}
\State Update $n_{t}(b_t^j) \leftarrow n_{t-1}(b_t^j) + 1$
\State Observe $\tilde{g}_t^j(b_t^j)$ and $\tilde{p}^j_t(b_t^j)$, then update
\State $$\widehat{g}_t(b_t^j) \leftarrow \left[\widehat{g}_{t-1}(b_t^j) \cdot n_{t-1}(b_t^j) + \widetilde{g}^j_t(b_t^j)\right]\big/n_t(b_t^j)$$
\State $$\widehat{p}_t(b_t^j) \leftarrow \left[\widehat{p}_{t-1}(b_t^j) \cdot n_{t-1}(b_t^j) + \widetilde{p}^j_t(b_t^j)\right]\big/n_t(b_t^j)$$
\EndFor
\For{$b\notin \bvec_t$}
\State $n_t(b) \leftarrow n_{t-1}(b)$
\EndFor
\EndFor
\end{algorithmic}
\caption{$\rgtucb$ Algorithm for unknown valuation.}\label{alg:rgtucb-multiple-bids-nonfixed-value}
\end{algorithm}
\footnotetext{We choose the remaining $m-1$ bids be different with each other and $b^1_t$.}

We denote $\Delta(v,b)\triangleq rgt(v^*, b^*) - rgt(v,b)$ and start with the following lemma that gives the concentration property of the $\mathrm{UCB}^\text{rgt}$ terms.
\begin{lemma}\label{lem:rgt-utility}
At iteration $t-1, (t\geq 2)$, for a (value, bid) pair $(v, b)\neq (v^*, b^*)$, where $v$ and $b$ are both observed at least $\frac{48(m+1)(2U)^2\ln t}{n\Delta(v, b)^2}$ times, then with probability at least $1 - \frac{4}{t^2}$, $\mathrm{UCB}^\text{rgt}_t(v, b)\leq \mathrm{UCB}^\text{rgt}_t(v^*, b^*)$.
\end{lemma}
\begin{proof}
First, we denote $C_t(v,b) = \frac{48(m+1)(2U)^2\ln t}{n\Delta(v, b)^2}$ and observe $\widehat{rgt}_t(v,b) \in [-2U, 2U]$, each auction is i.i.d sampled from unknown distribution and each block contains $\frac{n}{m+1}$ auctions. Based on union bounds, i.e., inequality $(u)$, and Hoeffding bounds, i.e., inequality $(h)$, we have for any positive integers $s,\ell \geq C_t(v,b) = \frac{48(m+1)(2U)^2\ln t}{n\Delta(v, b)^2}$,
\begin{equation}\label{ie:hoeffding3}
\begin{aligned}
&\PP\left(\widehat{rgt}_t(v, b) - rgt(v, b)\geq \frac{\Delta(v, b)}{2}\Big| n_{t-1}(b)=s, n_{t-1}(v)=\ell\right)\\
&~~~~\overset{(u)}{\leq} \PP\left(\widehat{u}_t(v, b) - u^*(v, b)\geq \frac{\Delta(v, b)}{4}\Big| n_{t-1}(b)=s\right)+ \PP\left(\widehat{u}_t(v, v) - u^*(v, v)\geq \frac{\Delta(v, b)}{4}\Big| n_{t-1}(v)=\ell\right)\\
&~~~~\overset{(h)}{\leq} \exp\left(\frac{-2n\cdot s\cdot\Delta(v, b)^2}{16(m+1)\cdot(2U)^2}\right) + \exp\left(\frac{-2n\cdot \ell\cdot\Delta(v, b)^2}{16(m+1)\cdot(2U)^2}\right)\leq \frac{2}{t^6},
\end{aligned}
\end{equation}
Moreover, for any $s^*, \ell^*\in \mathbb{N}_{+}$, we have
\begin{equation}\label{ie:hoeffding4}
\begin{aligned}
&\PP\left(rgt(v^*, b^*) - \widehat{rgt}_t(v^*, b^*)\geq 4U\sqrt{\frac{3(m+1)\ln t}{n\cdot \big(n_{t-1}(v^*)\wedge n_{t-1}(b^*)\big)}} \Big|n_{t-1}(b^*)=s^*, n_{t-1}(v^*)=\ell^*\right)\\
& \overset{(u)}{\leq} \PP\left(u^*(v^*, b^*) - \widehat{u}_t(v^*, b^*)\geq 2U\sqrt{\frac{3(m+1)\ln t}{n\cdot \big(n_{t-1}(v^*)\wedge n_{t-1}(b^*)\big)}}\Big|n_{t-1}(b^*)=s^*\right)\\
& \hspace{.5cm} + \PP\left(u(v^*, v^*) - \widehat{u}_t(v^*, v^*)\geq 2U\sqrt{\frac{3(m+1)\ln t}{n\cdot \big(n_{t-1}(v^*)\wedge n_{t-1}(b^*)\big)}}\Big|n_{t-1}(v^*)=\ell^*\right)\\
& \overset{(h)}{\leq} \exp(-6\ln t) + \exp(-6\ln t) = \frac{2}{t^6}
\end{aligned}
\end{equation}
If the inequalities~(\ref{ie:hoeffding3}) and~(\ref{ie:hoeffding4}) are both violated, then it is trivial to show, given $n_{t-1}(b)=s, n_{t-1}(v)=\ell, n_{t-1}(b^*)=s^*$ and $n_{t-1}(v^*)=\ell^*$, where $s,\ell \geq C_t(v,b)$, $\mathrm{UCB}^\text{rgt}_t(v, b) \leq \mathrm{UCB}^\text{rgt}_t(v^*, b^*)$, since
\begin{align*}
\mathrm{UCB}^\text{rgt}_t(v, b) &\leq \widehat{rgt}_t(v,b)+\frac{\Delta(v, b)}{2}\leq rgt(v, b) + \Delta(b) = rgt(v^*, b^*) \leq \mathrm{UCB}^\text{rgt}_t(v^*, b^*)
\end{align*}
where the first inequality holds because $$n_{t-1}(v)\wedge n_{t-1}(b)\geq \frac{48(m+1)(2U)^2\ln t}{n\Delta(v, b)^2},$$ the second inequality holds if inequality~(\ref{ie:hoeffding3}) is violated and the final inequality is based on the violation of inequality~(\ref{ie:hoeffding4}). Let event $\zeta_t(s,\ell,s^*,\ell^*)$ be $\big\{n_{t-1}(b)=s, n_{t-1}(v)=\ell, n_{t-1}(b^*)=s^*, n_{t-1}(v^*)=\ell^*\big\}$ and by union bound, we have
\begin{align*}
&\PP\left(\mathrm{UCB}^\text{rgt}_t(v, b) \geq \mathrm{UCB}^\text{rgt}_t(v^*, b^*)\right)\\
&\leq \sum_{s=C_t(v,b)}^{t-1}\sum_{\ell=C_t(v,b)}^{t-1}\sum_{s^*=1}^{t-1}\sum_{\ell^*=1}^{t-1}\PP\left(\mathrm{UCB}^\text{rgt}_t(v, b) \geq \mathrm{UCB}^\text{rgt}_t(v^*, b^*)\Big|\zeta_t(s, \ell, s^*, \ell^*)\right)\\
&\leq \sum_{s=C_t(v,b)}^{t-1}\sum_{\ell=C_t(v,b)}^{t-1}\sum_{s^*=1}^{t-1}\sum_{\ell^*=1}^{t-1}\frac{4}{t^6} \leq \frac{4}{t^2}
\end{align*}
\end{proof}

Utilizing Lemma~\ref{lem:rgt-utility}, we show the worst case {\em Pseudo-Regret} analysis for IC-testing with non-fixed valuation setting in Theorem~\ref{thm:rgt-bound-multiple-bids-nonfixed-value}, where the full proof is shown in Appendix~\ref{app:proof-thm2}.
\begin{theorem}\label{thm:rgt-bound-multiple-bids-nonfixed-value}
$\rgtucb$ algorithm for unknown valuation setting (DSP problem) achieves pseudo-regret at most \begin{align*}
&\hspace{.3cm}\smashoperator{\sum_{\substack{v, b\in B,\\(v,b)\neq(v^*, b^*)}}} \frac{384(m+1)U^2\ln T}{n\Delta(v,b)}+ \frac{2\pi^2\Delta(v,b)}{3} \cdot\left(\1\big\{v\neq v^*\big\} + \frac{\1\big\{v= v^*\big\}}{m}\right)\end{align*}
\end{theorem}

Following the same argument in Section~\ref{sec:fix-value}, if we choose $m$ appropriately, we will get the following asymptotic bound. 
\begin{corollary}\label{coro:non-fix-value-error}
Let $\overline{\Delta} \triangleq \max_{v, b\in B} \Delta(v, b)$ and $\underline{\Delta} \triangleq \min_{v, b\in B} \Delta(v, b)$. Choosing $m = \frac{\pi}{24U}\sqrt{\frac{n\overline{\Delta} \underline{\Delta}}{|B|\ln T}}$, the asymptotic error of determining IC regret achieved by $\rgtucb$ algorithm is upper bounded by
\begin{align*}
\mathcal{E}(T)  %
\leq O\left(\frac{|B|}{T}\cdot\left(\sqrt{\frac{|B|\ln T}{n}} + \frac{|B|\ln T}{n}\right)\right)
\end{align*}
\end{corollary}
Note in the DSP problem,  since we consider determining the worst-case IC regret over all possible valuations $v$, our asymptotic error bound quadratically grows as $|B|$, while for the advertiser problem the error linearly grows as $|B|$ since the benchmark we consider there is weaker than it in the DSP problem.

\section{Extensions and Discussions}
\subsection{Switching value and bid}
In Section~\ref{sec:non-fix-value}, we assume the test bidder (DSP) only chooses one valuation $v_t$ at every time step $t$ to empirically estimate IC regret. However, in practice, for each (value, bid) pair, the test bidder can switch the roles of bid and value to better estimate the IC regret.

Since there is no conceptual difference between bid and value in this setting, at time step $t$, the test bidder also submit a "bid" $b^{m+1}_t$ for every auction in block $m+1$. We still denote the bids submitted at time $t$ be $\bvec_t=\big\{b^1_t, b^2_t,\cdots, b^{m+1}_t\big\}$. The only difference between this case and our standard model is from the definition of {\em empirical IC regret}, we define the modified {\em empirical IC regret} in this setting as below,

\begin{definition}[modified empirical IC regret]
$$\widetilde{rgt}_t(\bvec_t) = \max_{i, j\in[m+1], i\neq j}\left\{\widetilde{u}^j_t (b_t^i, b^j_t) - \widetilde{u}^i_t (b^i_t, b^i_t)\right\}$$
\end{definition}

where we consider all possible (value, bid) pair from $b_t$. The learning task in this setting is to design an efficient algorithm to generate $m+1$ bids at every time step to minimize {\em Pseudo-Regret}, $$\E\left[R(T)\right]\triangleq  \max_{v, b\in B}\sum_{t=1}^T rgt(v, b) - \E\left[\sum_{t=1}^T\widetilde{rgt}_t(\bvec_t)\right]$$ We extend our $\rgtucb$ algorithm with fixed valuation (Algorithm~\ref{alg:rgtucb-multiple-bids}) for this setting and show the pseudo-code in Algorithm~\ref{alg:rgtucb-multiple-bids-switching}. Very similar to Algorithm~\ref{alg:rgtucb-multiple-bids}, at each time step, we update the $\mathrm{UCB}^\text{rgt}$ terms defined in Equation~\ref{eq:ucb-rgt} for every $(b_1, b_2)$ bids pair. Given these $\mathrm{UCB}^\text{rgt}$ terms, we run $\genbids$ algorithm proposed in Algorithm ~\ref{alg:generate-bids} to generate $m+1$ bids, which aims to select the (value, bid) pairs with the largest $\mathrm{UCB}^{\text{rgt}}$ terms.

\begin{algorithm}[H]
\begin{algorithmic}
\State \textbf{Input:} A finite set of bids $B$, parameter $m, n$. $\forall b\in B, n_0(b)=1$
\State \textbf{Initialize:} Run $\init(B, m,n)$ algorithm to get $\hat{g}_0$ and $\hat{p}_0$. 
\For{$t = 1,\cdots, T$}
\State Update $\mathrm{UCB}$ terms on the expected IC regret of every $(b_1, b_2)$ bids pair
\State $\mathrm{UCB}^\text{rgt}_t(b_1, b_2) = \widehat{rgt}_t(b_1, b_2) + 4U\sqrt{\frac{3(m+1)\ln t}{n\cdot \big(n_{t-1}(b_1)\wedge n_{t-1}(b_2)\big)}}$
\State Run $\genbids$ Algorithm on $\mathrm{UCB}^\text{rgt}_t$ terms to generate $m+1$ bids, which forms $\bvec_t=\big\{b_t^1,\cdots,b_t^{m+1}\big\}$, where $\genbids$ is defined in Algorithm~\ref{alg:generate-bids}.
\For{$j = 1,\cdots,m+1$}
\State Update $n_t(b_t^j)\leftarrow n_{t-1}(b_t^j) + 1$
\State Observe $\tilde{g}_t^j(b_t^j)$ and $\tilde{p}^j_t(b_t^j)$
\State $$\hat{g}_t(b_t^j) \leftarrow \left[\hat{g}_{t-1}(b_t^j) \cdot n_{t-1}(b_t^j) + \tilde{g}^j_t(b_t^j)\right]\big/n_t(b_t^j)$$
\State $$\hat{p}_t(b_t^j) \leftarrow \left[\hat{p}_{t-1}(b_t^j) \cdot n_{t-1}(b_t^j) + \tilde{p}^j_t(b_t^j)\right]\big/n_t(b_t^j)$$
\EndFor
\For{$b\notin \bvec_t$}
\State $n_t(b) \leftarrow n_{t-1}(b)$, $\hat{g}_t(b)\leftarrow \hat{g}_{t-1}(b)$, $\hat{p}_t(b)\leftarrow \hat{p}_{t-1}(b)$
\EndFor
\EndFor
\end{algorithmic}
\caption{$\rgtucb$ Algorithm for IC testing when allowing switching value and bid.}\label{alg:rgtucb-multiple-bids-switching}
\end{algorithm}

\begin{algorithm}[H]
\begin{algorithmic}
\State \textbf{Input:} A finite set of bids $B$, $\mathrm{UCB}^\text{rgt}$ terms for every $(b_1, b_2)$ bids pair, where $b_1, b_2\in B$. Initialize a set $B_0=\{ \}$.
\While{$|B_0| < m+1$}
\State Choose the largest $\mathrm{UCB}^\text{rgt}(b_1, b_2)$ such that $b_1\notin B_0$ or $b_2\notin B_0$.
\State Update $B_0 \leftarrow B_0\cup \{b_1, b_2\}$. \footnotemark
\EndWhile
\end{algorithmic}
\caption{$\genbids$ Algorithm for generating bids given $\mathrm{UCB}^\text{rgt}$ terms.}\label{alg:generate-bids}
\end{algorithm}
\footnotetext{If the two bids chosen in the last round are both not in $B_0$, we randomly choose one to add it in $B_0$.}

We show the {\em Pseudo-Regret} guarantee for this setting in Theorem~\ref{thm:rgt-bound-switching}, the proof is in Appendix~\ref{app:proof-thm-switching}.

\begin{theorem}\label{thm:rgt-bound-switching}
$\rgtucb$ algorithm for DSP problem when allowing switching value and bid achieves pseudo-regret at most $$\sum_{\substack{b_1, b_2\in B,\\ (b_1, b_2) \neq (v^*, b^*)}} \frac{192(m+1)U^2\ln T}{n\Delta(b_1,b_2)} +  \frac{2\pi^2\Delta(b_1,b_2)}{3(m+1)^2}$$
where $\Delta(b_1, b_2)=rgt(v^*, b^*) - rgt(b_1, b_2)$.
\end{theorem}

Note the {\em Pseudo-Regret} bound achieved by Algorithm~\ref{alg:rgtucb-multiple-bids-switching} is always better than it achieved by Algorithm~\ref{alg:rgtucb-multiple-bids-nonfixed-value} for any $m$. Then we can always get a better asymptotic error bound by Algorithm~\ref{alg:rgtucb-multiple-bids-switching} for the DSP problem. %
\begin{corollary}
Let $\overline{\Delta} \triangleq \max_{v, b\in B} \Delta(v, b)$ and $\underline{\Delta} \triangleq \min_{v, b\in B} \Delta(v, b)$. Choosing $m+1 = O\left(\frac{n}{\ln T}\right)^{1/3}$ the asymptotic error of determining IC regret achieved by $\rgtucb$ algorithm is upper bounded by\footnote{This asymptotic error bound can be achieved if the optimal $m$ can be chosen to satisfy the constraint that $m < |B|$.}
\begin{align*}
\mathcal{E}(T) &\leq O\left(\frac{|B|^2(\ln T)^{2/3}}{n^{2/3}T}\right)
\end{align*}
\end{corollary}

\subsection{Continuous bids space}\label{sec:continuous-bid-space}
Thus far, we have considered a discrete bid space. We now extend our framework to allow the advertiser to bid in a continuous bid space $\mathcal B$ (e.g. a uniform interval $[0, 1]$). Utilizing the discretization result in \cite{K05, KSU08, FPS18}, let $B$ be the discretization of continuous bid space $\mathcal B$ and $\mathcal{DE}(B,\mathcal B)$ to represent the discretization error of bid space $B$ and $\mathcal B$ such that $\mathcal{DE}(B, \mathcal B) = \sup_{v,b \in \mathcal B}\sum_{t=1}^T rgt(v, b) - \sup_{v',b' \in B}\sum_{t=1}^T rgt(v', b')$. To avoid the conflict, we denote the {\em Pseudo-Regret} in $\mathcal B$ be $\E\left[R(T, \mathcal B)\right]$ and the {\em Pseudo-Regret} in $B$ be $\E\left[R(T, B)\right]$. The following lemma uncovers the relationship of $\E\left[R(T, \mathcal B)\right]$ and $\E\left[R(T, B)\right]$
\begin{lemma}[\cite{K05,KSU08}]
Let $B$ be a discretization of a continuous bid space $\mathcal B$, then
\begin{align*}
\E[R(T, \mathcal B)] \leq \E[R(T, B)] + \mathcal{DE}(B, \mathcal B)
\end{align*}
\end{lemma}
In practice, the expected allocation function $g^*$ and payment function $p^*$ are both relatively smooth (see e.g.. the plots in \cite{LMSV18}). Assume the Lipschitzness of expected allocation function and payment function, the discretization error $\mathcal{DE}(B, \mathcal B)$ can be easily bounded, then our {\em Pseudo-Regret} analysis can be directly applied in this continuous bid space. %
 
\section{Simulations}\label{sec:simulation}
In this section, we run some simulations with Generalized Second Price (GSP) auctions to validate the theoretical results mentioned in previous sections. 
To visualize the performance of the algorithms, we plot the {\em Pseudo-Regret} versus the number of time steps\footnote{In the experiments section, {\em Pseudo-Regret} refers to the upper bound of $\E[R(T)]$ derived by Equations~(\ref{eq:ub-pseudo-regret1}),~(\ref{eq:pseudo-rgt-2}) and (\ref{eq:rgt-bound-3}) respectively.} . Every experiment is repeated for 10 times and the dash area represents the 95\% confidence intervals of {\em Pseudo-Regret} at each time step. \es{What do we want to better understand by running these simulations?}
\vspace{-5pt}
\subsection{Settings}
At each time $t$, the test bidder randomly participates $n$ GSP auctions. Each auction has 5 slots and 20 bidders without counting the test bidder, each bidder's bid is generated from uniform distribution on $[0,10]$. The {\em click-through-rate} sequence of each auction is generated from descending ordered $Beta(2,5)$ distribution (i.e. {\em click-through-rate} of each slot is first i.i.d generated from $Beta(2, 5)$ and then arranged by descending order). We consider a finite bids space $B \triangleq \{0.01, 0.02,\cdots, 10\}$.
\vspace{-5pt}
\subsection{The advertiser problem}
We test the performance of $\rgtucb$ algorithm with known valuation $v=9.5$. First, we fix the number of blocks be $16$ (i.e. $m=15$) and plot the {\em Pseudo-Regret} achieved by $\rgtucb$ with different number of auctions the test bidder participates at each time step, i.e. $n=16, 64, 256, 1024$ (Figure~\ref{fig:rgt-fix-value}(a)). We observe that the {\em Pseudo-Regret} decays linearly with $n$ which is consistent with our analysis. 

Second, we fix $n=1024$ and test the performance of the algorithm with different number of blocks, such as $m=1, 3, 7, 15, 31, 63$ (Figure~\ref{fig:rgt-fix-value}(a)). We find $m=7$ (8 blocks) achieves the lowest {\em Pseudo-Regret}. When $m$ is too large, the {\em Pseudo-Regret} curve incurs some shocks and suffers high variance because of the noisy observed information at each time. 
\begin{figure}
\begin{subfigure}{0.49\textwidth}
\centering
\includegraphics[scale=0.55]{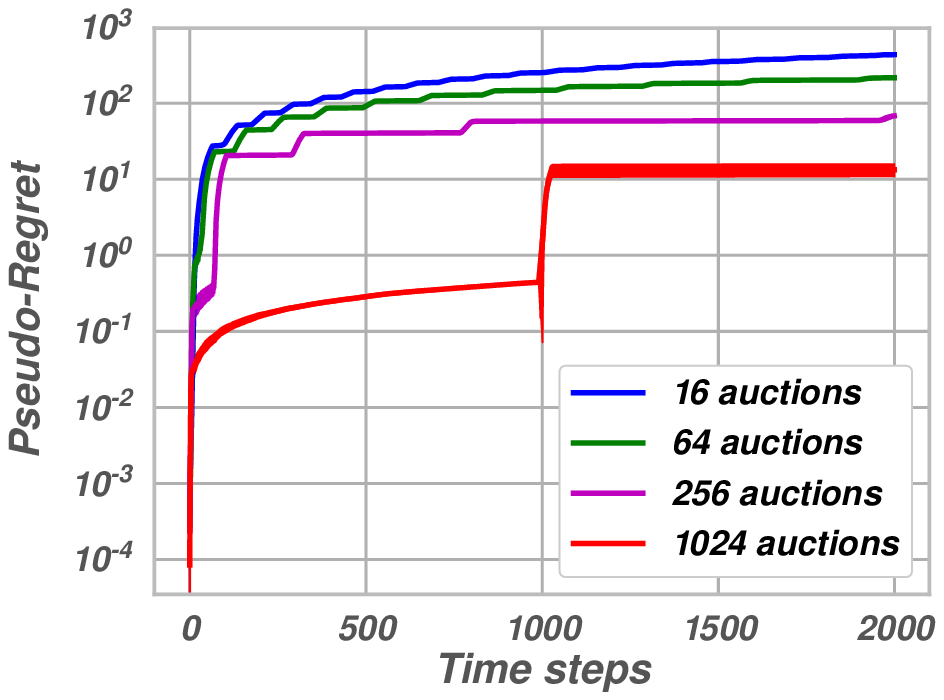}
~\\[2.5pt]\scriptsize{(a)}\vspace*{-5pt}
\end{subfigure}
\begin{subfigure}{0.49\textwidth}
\centering
\includegraphics[scale=0.55]{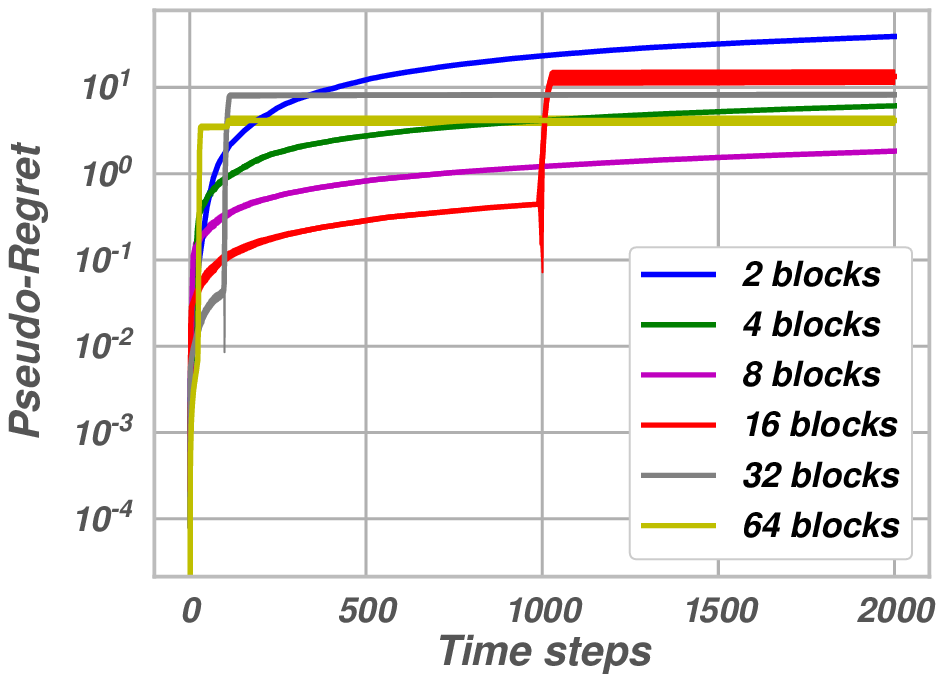}
~\\[2.5pt]\scriptsize{(b)}\vspace*{-5pt}
\end{subfigure}
\caption{A semi-logarithmic plot of {\em Pseudo-Regret} of $\rgtucb$ algorithm for known valuation $v=9.5$ (Algorithm~\ref{alg:rgtucb-multiple-bids}).}
\vspace{-10pt}
\label{fig:rgt-fix-value}
\end{figure}
\vspace{-5pt}
\subsection{The DSP problem}
In this simulation, we test the performance of $\rgtucb$ for unknown valuation case (i.e.,  the DSP problem). Similarly, we show the {\em Pseudo-Regret} curve for different $n$ given $m=15$ and different $m$ given $n=1024$ in Figure~\ref{fig:rgt-non-fix-value}. Figure~\ref{fig:rgt-non-fix-value}(a) validates that {\em Pseudo-Regret} decays when $n$ grows and we observe $m=3$ (4 blocks) gives the best {\em Pseudo-Regret}. This corresponds to our theory that the optimal $m$ in the DSP problem should be smaller than the optimal $m$ in the advertiser problem for the same auction setting when $|B|$ is large (see Corollary~(\ref{coro:fix-value-error}) and (\ref{coro:non-fix-value-error})).
\begin{figure} 
\begin{subfigure}{0.32\textwidth}
\centering
\includegraphics[scale=0.55]{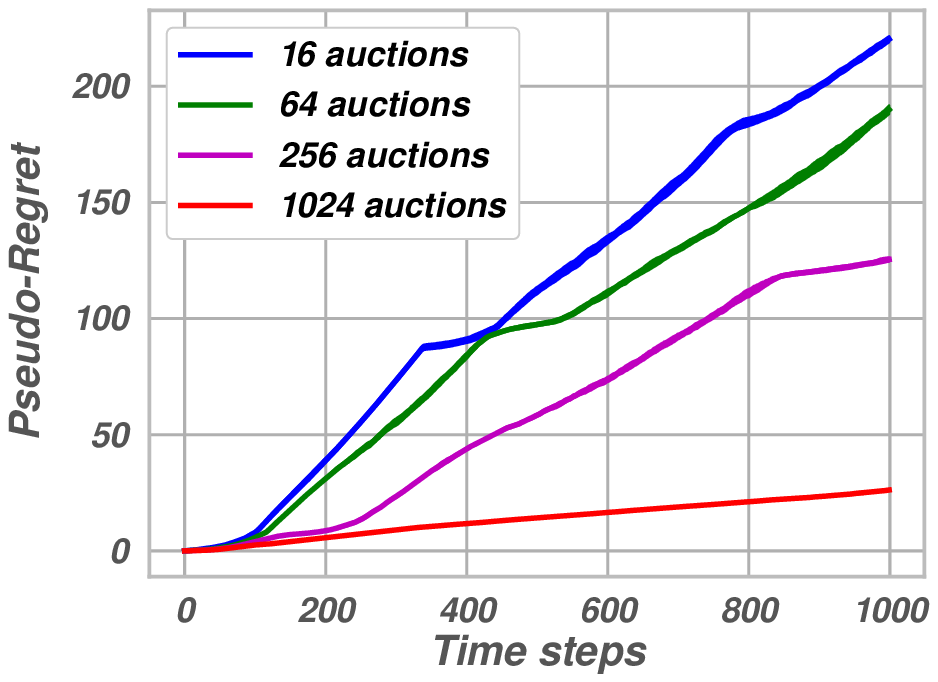}
~\\[2.5pt]\scriptsize{(a)}\vspace*{-5pt}
\end{subfigure}
\begin{subfigure}{0.32\textwidth}
\centering
\includegraphics[scale=0.55]{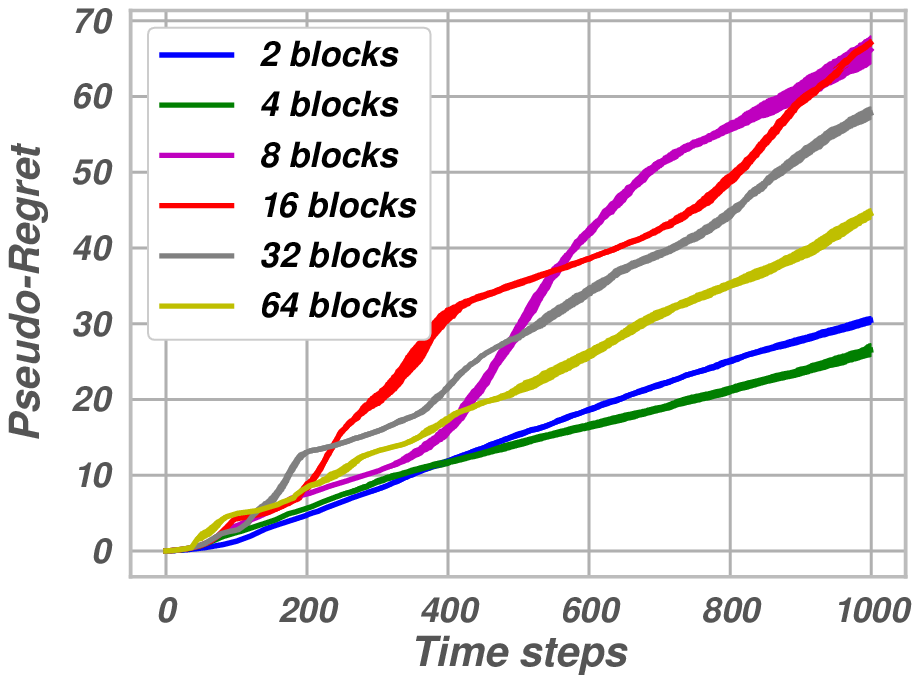}
~\\[2.5pt]\scriptsize{(b)}\vspace*{-5pt}
\end{subfigure}
\begin{subfigure}{0.32\textwidth}
\centering
\includegraphics[scale=0.55]{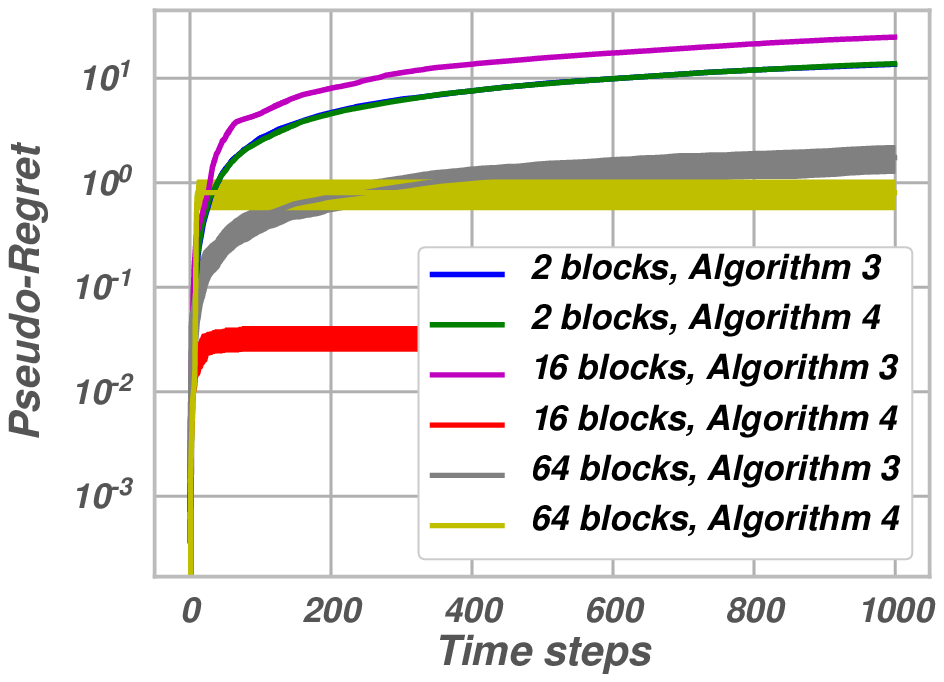}
~\\[2.5pt]\scriptsize{(c)}\vspace*{-5pt}
\end{subfigure}
\caption{{\em Pseudo-Regret} plot of $\rgtucb$ algorithm for unknown valuation. $(a)$: Algorithm~\ref{alg:rgtucb-multiple-bids-nonfixed-value} for different $n$, $(b)$: Algorithm~\ref{alg:rgtucb-multiple-bids-nonfixed-value} for different $m$ and $(c)$: a semi-logarithmic {\em Pseudo-Regret}  to compare Algorithm~\ref{alg:rgtucb-multiple-bids-nonfixed-value} and Algorithm~\ref{alg:rgtucb-multiple-bids-switching}. %
}
\vspace{-10pt}
\label{fig:rgt-non-fix-value}
\end{figure}
\vspace{-5pt}
\subsection{Efficiency of $\rgtucb$ algorithm}
We first introduce two standard baselines  used in this section, $\randombids$ algorithm and $\egreedy$ algorithm.

\textbf{$\randombids$. } At each time step $t$, the advertiser uniformly randomly choose $m$ bids $b_t$, while the DSP uniformly randomly choose $m$ bids $b_t$ and a valuation $v_t$ from $B$.
 
\textbf{$\egreedy$.} For the advertiser problem with known valuation $v$, at each time step, the advertiser uniformly randomly chooses $m$ bids from $B$ with probability $\epsilon$,  otherwise, chooses the $m$ bids that correspond to the largest average utility $\hat{u}(v, \cdot)$ terms.\footnote{In the experiments, we fix $\epsilon=0.1$} For the DSP problem, at each time step $t$, the DSP uniformly randomly choose $m$ bids $b_t$ and a valuation $v_t$ with probability $\epsilon$. Otherwise, the DSP chooses the $(v_t, b_t^1)$ pair corresponds to the largest $\widehat{rgt}$ terms and the rest $m-1$ bids associated with the largest $\hat{u}(v_t, \cdot)$ terms.

We compare the performance of our $\rgtucb$ algorithm with the above two baselines for the advertiser problem and the DSP problem. For the both settings, $\rgtucb$ algorithm performs better than two baselines for different number of blocks.
\begin{figure}
\begin{subfigure}{0.49\textwidth}
\centering
\includegraphics[scale=0.55]{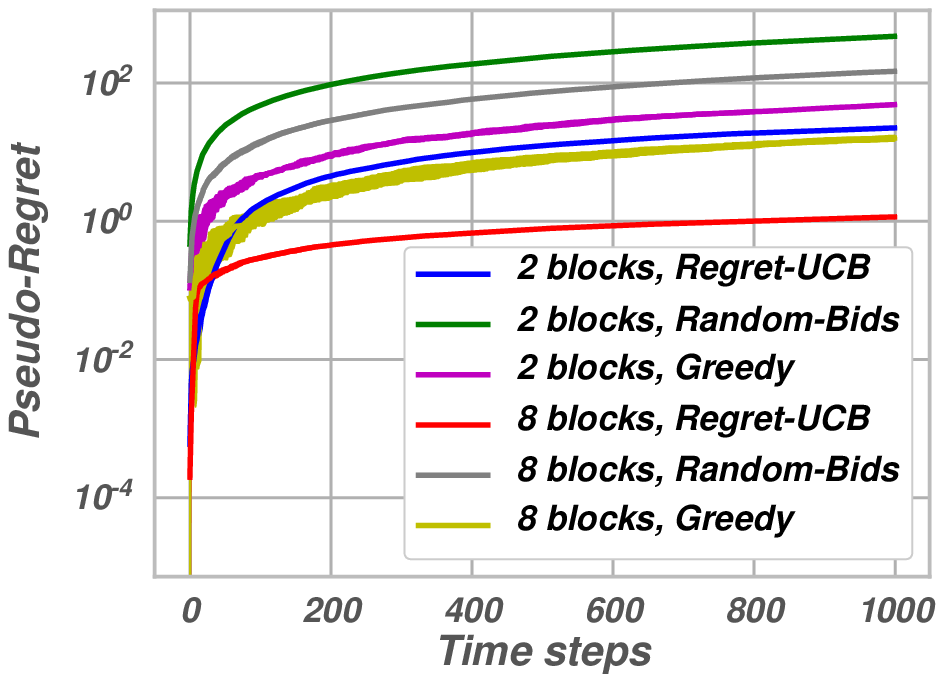}
~\\[2.5pt]\scriptsize{(a)}\vspace*{-5pt}
\end{subfigure}
\begin{subfigure}{0.49\textwidth}
\centering
\includegraphics[scale=0.55]{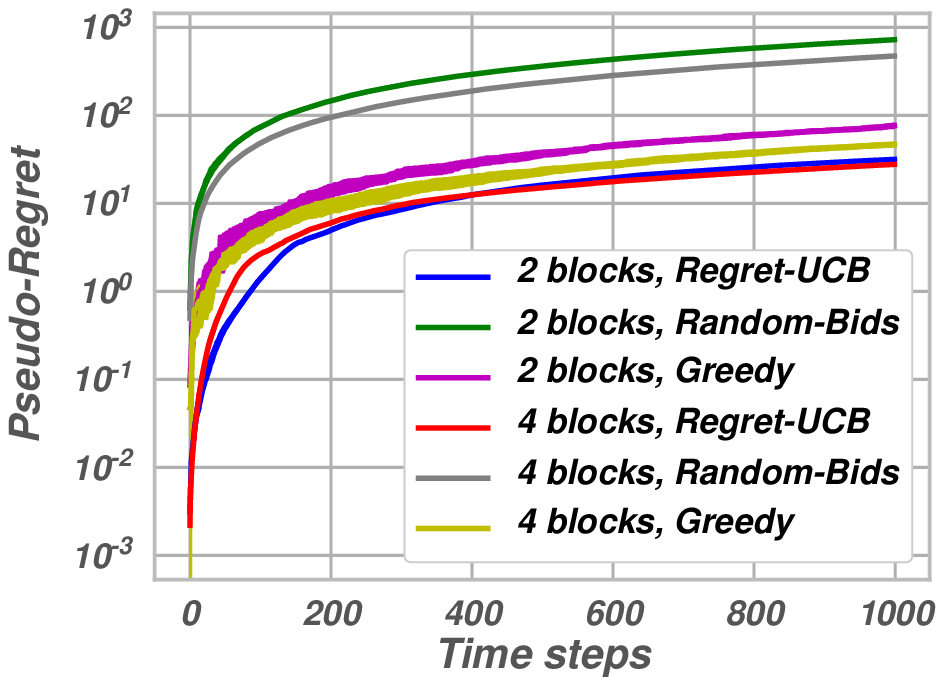}
~\\[2.5pt]\scriptsize{(b)}\vspace*{-5pt}
\end{subfigure}
\caption{A semi-logarithmic plot of {\em Pseudo-Regret} with $\randombids$, $\egreedy$ and $\rgtucb$ (Algorithm~\ref{alg:rgtucb-multiple-bids-nonfixed-value}) for the advertiser problem (a) and DSP problem (b).}
\vspace{-10pt}
\label{fig:rgt-baseline}
\end{figure}
\subsection{$\rgtucb$ when allowing switching value and bid}
Moreover, we compare the performance of Algorithm~\ref{alg:rgtucb-multiple-bids-nonfixed-value} and Algorithm~\ref{alg:rgtucb-multiple-bids-switching} for different $m$ given $n=1024$ in Figure~\ref{fig:rgt-non-fix-value}(c) to validate the theory that allowing switching value and bid leads to lower  Pseudo-Regret.\footnote{Based on our observation in the experiments, sometimes switching value and bid may "over-estimate" the IC regret, i.e. the empirical estimation of IC regret is sometimes larger than worst case IC regret.} In addition, we also observe the tradeoff between Pseudo-Regret and $m$ of Algorithm~\ref{alg:rgtucb-multiple-bids-switching}. Among $m=1, 15, 63$, the optimal $m$ is $15$ (16 blocks) for Algorithm~\ref{alg:rgtucb-multiple-bids-switching}. If we choose $m$ too large, like $m=63$, Pseudo-Regret is worse than $m=1$ case and incurs high variance.

\es{Need a conclusion.}
\section{Conclusion}
In this paper, we analyze the problem of measuring end-to-end Incentive Compatibility regret only given a black-box access to an auction mechanism. We consider two different main problems, where in the \emph{advertiser problem} the goal is to measure IC regret for some known valuation $v$ and in the more general \emph{demand-side platform (DSP) problem} we want to determine the worst-case IC regret over all possible valuations. For both problems, we propose online learning algorithms to measure the IC regret given the feedback (allocation and payment) of the auction mechanism when the test bidder participating auctions at each time. We theoretically bound the error of determining IC regret for both problems and extend to more general settings. Finally, we run simulations to validate our results. For next steps, it is worth exploring further avenues of relaxing the informational assumptions (e.g. what if the distributions of allocation and payment of auctions are non stationary) and testing our algorithms over real auctions data.

\bibliographystyle{ACM-Reference-Format}
\bibliography{ic-bandit}


\begin{thebibliography}{23}

%%% the Web address.

\ifx \showCODEN    \undefined \def \showCODEN     #1{\unskip}     \fi
\ifx \showDOI      \undefined \def \showDOI       #1{#1}\fi
\ifx \showISBNx    \undefined \def \showISBNx     #1{\unskip}     \fi
\ifx \showISBNxiii \undefined \def \showISBNxiii  #1{\unskip}     \fi
\ifx \showISSN     \undefined \def \showISSN      #1{\unskip}     \fi
\ifx \showLCCN     \undefined \def \showLCCN      #1{\unskip}     \fi
\ifx \shownote     \undefined \def \shownote      #1{#1}          \fi
\ifx \showarticletitle \undefined \def \showarticletitle #1{#1}   \fi
\ifx \showURL      \undefined \def \showURL       {\relax}        \fi
% invisible to TeX
\providecommand\bibfield[2]{#2}
\providecommand\bibinfo[2]{#2}
\providecommand\natexlab[1]{#1}
\providecommand\showeprint[2][]{arXiv:#2}

\bibitem[\protect\citeauthoryear{Agarwal, Hsu, Kale, Langford, Li, and
  Schapire}{Agarwal et~al\mbox{.}}{2014}]%
        {Agarwal14}
\bibfield{author}{\bibinfo{person}{Alekh Agarwal}, \bibinfo{person}{Daniel
  Hsu}, \bibinfo{person}{Satyen Kale}, \bibinfo{person}{John Langford},
  \bibinfo{person}{Lihong Li}, {and} \bibinfo{person}{Robert Schapire}.}
  \bibinfo{year}{2014}\natexlab{}.
\newblock \showarticletitle{Taming the Monster: A Fast and Simple Algorithm for
  Contextual Bandits}. In \bibinfo{booktitle}{\emph{Proceedings of the 31st
  International Conference on Machine Learning}}
  \emph{(\bibinfo{series}{Proceedings of Machine Learning Research})},
  \bibfield{editor}{\bibinfo{person}{Eric~P. Xing} {and} \bibinfo{person}{Tony
  Jebara}} (Eds.), Vol.~\bibinfo{volume}{32}. \bibinfo{publisher}{PMLR},
  \bibinfo{address}{Bejing, China}, \bibinfo{pages}{1638--1646}.
\newblock


\bibitem[\protect\citeauthoryear{Amin, Cummings, Dworkin, Kearns, and
  Roth}{Amin et~al\mbox{.}}{2015}]%
        {ACDKR15}
\bibfield{author}{\bibinfo{person}{Kareem Amin}, \bibinfo{person}{Rachel
  Cummings}, \bibinfo{person}{Lili Dworkin}, \bibinfo{person}{Michael Kearns},
  {and} \bibinfo{person}{Aaron Roth}.} \bibinfo{year}{2015}\natexlab{}.
\newblock \showarticletitle{Online Learning and Profit Maximization from
  Revealed Preferences}. In \bibinfo{booktitle}{\emph{AAAI}}.
  \bibinfo{pages}{770--776}.
\newblock


\bibitem[\protect\citeauthoryear{Amin, Rostamizadeh, and Syed}{Amin
  et~al\mbox{.}}{2014}]%
        {ARS14}
\bibfield{author}{\bibinfo{person}{Kareem Amin}, \bibinfo{person}{Afshin
  Rostamizadeh}, {and} \bibinfo{person}{Umar Syed}.}
  \bibinfo{year}{2014}\natexlab{}.
\newblock \showarticletitle{Repeated contextual auctions with strategic
  buyers}. In \bibinfo{booktitle}{\emph{Advances in Neural Information
  Processing Systems}}. \bibinfo{pages}{622--630}.
\newblock


\bibitem[\protect\citeauthoryear{Balseiro and Gur}{Balseiro and Gur}{2017}]%
        {BG17}
\bibfield{author}{\bibinfo{person}{Santiago Balseiro} {and}
  \bibinfo{person}{Yonatan Gur}.} \bibinfo{year}{2017}\natexlab{}.
\newblock \showarticletitle{Learning in Repeated Auctions with Budgets: Regret
  Minimization and Equilibrium}.
\newblock  (\bibinfo{year}{2017}).
\newblock


\bibitem[\protect\citeauthoryear{Blum, Kumar, Rudra, and Wu}{Blum
  et~al\mbox{.}}{2004}]%
        {BKRW04}
\bibfield{author}{\bibinfo{person}{Avrim Blum}, \bibinfo{person}{Vijay Kumar},
  \bibinfo{person}{Atri Rudra}, {and} \bibinfo{person}{Felix Wu}.}
  \bibinfo{year}{2004}\natexlab{}.
\newblock \showarticletitle{Online learning in online auctions}.
\newblock \bibinfo{journal}{\emph{Theoretical Computer Science}}
  \bibinfo{volume}{324}, \bibinfo{number}{2-3} (\bibinfo{year}{2004}),
  \bibinfo{pages}{137--146}.
\newblock


\bibitem[\protect\citeauthoryear{Bubeck, Cesa-Bianchi, et~al\mbox{.}}{Bubeck
  et~al\mbox{.}}{2012}]%
        {BCB2012}
\bibfield{author}{\bibinfo{person}{S{\'e}bastien Bubeck},
  \bibinfo{person}{Nicolo Cesa-Bianchi}, {et~al\mbox{.}}}
  \bibinfo{year}{2012}\natexlab{}.
\newblock \showarticletitle{Regret analysis of stochastic and nonstochastic
  multi-armed bandit problems}.
\newblock \bibinfo{journal}{\emph{Foundations and Trends{\textregistered} in
  Machine Learning}} \bibinfo{volume}{5}, \bibinfo{number}{1}
  (\bibinfo{year}{2012}), \bibinfo{pages}{1--122}.
\newblock


\bibitem[\protect\citeauthoryear{Chawla, Hartline, and Nekipelov}{Chawla
  et~al\mbox{.}}{2014}]%
        {CHN14}
\bibfield{author}{\bibinfo{person}{Shuchi Chawla}, \bibinfo{person}{Jason~D.
  Hartline}, {and} \bibinfo{person}{Denis Nekipelov}.}
  \bibinfo{year}{2014}\natexlab{}.
\newblock \showarticletitle{Mechanism design for data science}. In
  \bibinfo{booktitle}{\emph{{ACM} Conference on Economics and Computation, {EC}
  '14, Stanford , CA, USA, June 8-12, 2014}}. \bibinfo{pages}{711--712}.
\newblock


\bibitem[\protect\citeauthoryear{Chen, Hu, Li, Li, Liu, and Lu}{Chen
  et~al\mbox{.}}{2016}]%
        {CHLLLL16}
\bibfield{author}{\bibinfo{person}{Wei Chen}, \bibinfo{person}{Wei Hu},
  \bibinfo{person}{Fu Li}, \bibinfo{person}{Jian Li}, \bibinfo{person}{Yu Liu},
  {and} \bibinfo{person}{Pinyan Lu}.} \bibinfo{year}{2016}\natexlab{}.
\newblock \showarticletitle{Combinatorial multi-armed bandit with general
  reward functions}. In \bibinfo{booktitle}{\emph{Advances in Neural
  Information Processing Systems}}. \bibinfo{pages}{1659--1667}.
\newblock


\bibitem[\protect\citeauthoryear{Chen, Wang, and Yuan}{Chen
  et~al\mbox{.}}{2013}]%
        {CWY13}
\bibfield{author}{\bibinfo{person}{Wei Chen}, \bibinfo{person}{Yajun Wang},
  {and} \bibinfo{person}{Yang Yuan}.} \bibinfo{year}{2013}\natexlab{}.
\newblock \showarticletitle{Combinatorial multi-armed bandit: General framework
  and applications}. In \bibinfo{booktitle}{\emph{International Conference on
  Machine Learning}}. \bibinfo{pages}{151--159}.
\newblock


\bibitem[\protect\citeauthoryear{Clarke}{Clarke}{1971}]%
        {C71}
\bibfield{author}{\bibinfo{person}{Edward~H Clarke}.}
  \bibinfo{year}{1971}\natexlab{}.
\newblock \showarticletitle{Multipart pricing of public goods}.
\newblock \bibinfo{journal}{\emph{Public choice}} \bibinfo{volume}{11},
  \bibinfo{number}{1} (\bibinfo{year}{1971}), \bibinfo{pages}{17--33}.
\newblock


\bibitem[\protect\citeauthoryear{Conitzer, Kroer, Sodomka, and Moses}{Conitzer
  et~al\mbox{.}}{2017}]%
        {CKSS18}
\bibfield{author}{\bibinfo{person}{Vincent Conitzer},
  \bibinfo{person}{Christian Kroer}, \bibinfo{person}{Eric Sodomka}, {and}
  \bibinfo{person}{Nicol{\'{a}}s E.~Stier Moses}.}
  \bibinfo{year}{2017}\natexlab{}.
\newblock \showarticletitle{Multiplicative Pacing Equilibria in Auction
  Markets}.
\newblock \bibinfo{journal}{\emph{CoRR}}  \bibinfo{volume}{abs/1706.07151}
  (\bibinfo{year}{2017}).
\newblock
\showeprint[arxiv]{1706.07151}


\bibitem[\protect\citeauthoryear{Feng, Podimata, and Syrgkanis}{Feng
  et~al\mbox{.}}{2018}]%
        {FPS18}
\bibfield{author}{\bibinfo{person}{Zhe Feng}, \bibinfo{person}{Chara Podimata},
  {and} \bibinfo{person}{Vasilis Syrgkanis}.} \bibinfo{year}{2018}\natexlab{}.
\newblock \showarticletitle{Learning to Bid Without Knowing Your Value}. In
  \bibinfo{booktitle}{\emph{Proceedings of the 2018 ACM Conference on Economics
  and Computation}} \emph{(\bibinfo{series}{EC '18})}.
  \bibinfo{publisher}{ACM}, \bibinfo{address}{New York, NY, USA},
  \bibinfo{pages}{505--522}.
\newblock
\showISBNx{978-1-4503-5829-3}


\bibitem[\protect\citeauthoryear{Gai, Krishnamachari, and Jain}{Gai
  et~al\mbox{.}}{2012}]%
        {Gai2012}
\bibfield{author}{\bibinfo{person}{Y. Gai}, \bibinfo{person}{B.
  Krishnamachari}, {and} \bibinfo{person}{R. Jain}.}
  \bibinfo{year}{2012}\natexlab{}.
\newblock \showarticletitle{Combinatorial Network Optimization With Unknown
  Variables: Multi-Armed Bandits With Linear Rewards and Individual
  Observations}.
\newblock \bibinfo{journal}{\emph{IEEE/ACM Transactions on Networking}}
  \bibinfo{volume}{20}, \bibinfo{number}{5} (\bibinfo{date}{Oct}
  \bibinfo{year}{2012}), \bibinfo{pages}{1466--1478}.
\newblock
\showISSN{1063-6692}


\bibitem[\protect\citeauthoryear{Groves}{Groves}{1973}]%
        {G73}
\bibfield{author}{\bibinfo{person}{Theodore Groves}.}
  \bibinfo{year}{1973}\natexlab{}.
\newblock \showarticletitle{Incentives in teams}.
\newblock \bibinfo{journal}{\emph{Econometrica: Journal of the Econometric
  Society}} (\bibinfo{year}{1973}), \bibinfo{pages}{617--631}.
\newblock


\bibitem[\protect\citeauthoryear{Kleinberg, Slivkins, and Upfal}{Kleinberg
  et~al\mbox{.}}{2008}]%
        {KSU08}
\bibfield{author}{\bibinfo{person}{Robert Kleinberg},
  \bibinfo{person}{Aleksandrs Slivkins}, {and} \bibinfo{person}{Eli Upfal}.}
  \bibinfo{year}{2008}\natexlab{}.
\newblock \showarticletitle{Multi-armed bandits in metric spaces}. In
  \bibinfo{booktitle}{\emph{Proceedings of the fortieth annual ACM symposium on
  Theory of computing}}. ACM, \bibinfo{pages}{681--690}.
\newblock


\bibitem[\protect\citeauthoryear{Kleinberg}{Kleinberg}{2005}]%
        {K05}
\bibfield{author}{\bibinfo{person}{Robert~D Kleinberg}.}
  \bibinfo{year}{2005}\natexlab{}.
\newblock \showarticletitle{Nearly tight bounds for the continuum-armed bandit
  problem}. In \bibinfo{booktitle}{\emph{Advances in Neural Information
  Processing Systems}}. \bibinfo{pages}{697--704}.
\newblock


\bibitem[\protect\citeauthoryear{Kveton, Wen, Ashkan, and Szepesvari}{Kveton
  et~al\mbox{.}}{2015}]%
        {KWAS15}
\bibfield{author}{\bibinfo{person}{Branislav Kveton}, \bibinfo{person}{Zheng
  Wen}, \bibinfo{person}{Azin Ashkan}, {and} \bibinfo{person}{Csaba
  Szepesvari}.} \bibinfo{year}{2015}\natexlab{}.
\newblock \showarticletitle{Tight regret bounds for stochastic combinatorial
  semi-bandits}. In \bibinfo{booktitle}{\emph{Artificial Intelligence and
  Statistics}}. \bibinfo{pages}{535--543}.
\newblock


\bibitem[\protect\citeauthoryear{Lahaie, Munoz~Medina, Sivan, and
  Vassilvitskii}{Lahaie et~al\mbox{.}}{2018}]%
        {LMSV18}
\bibfield{author}{\bibinfo{person}{S{\'e}bastien Lahaie},
  \bibinfo{person}{Andr{\'e}s Munoz~Medina}, \bibinfo{person}{Balasubramanian
  Sivan}, {and} \bibinfo{person}{Sergei Vassilvitskii}.}
  \bibinfo{year}{2018}\natexlab{}.
\newblock \showarticletitle{Testing Incentive Compatibility in Display Ad
  Auctions}. In \bibinfo{booktitle}{\emph{Proceedings of the 2018 World Wide
  Web Conference}} \emph{(\bibinfo{series}{WWW '18})}.
  \bibinfo{publisher}{International World Wide Web Conferences Steering
  Committee}, \bibinfo{address}{Republic and Canton of Geneva, Switzerland},
  \bibinfo{pages}{1419--1428}.
\newblock
\showISBNx{978-1-4503-5639-8}


\bibitem[\protect\citeauthoryear{Myerson}{Myerson}{1981}]%
        {M81}
\bibfield{author}{\bibinfo{person}{Roger~B Myerson}.}
  \bibinfo{year}{1981}\natexlab{}.
\newblock \showarticletitle{Optimal auction design}.
\newblock \bibinfo{journal}{\emph{Mathematics of operations research}}
  \bibinfo{volume}{6}, \bibinfo{number}{1} (\bibinfo{year}{1981}),
  \bibinfo{pages}{58--73}.
\newblock


\bibitem[\protect\citeauthoryear{Nisan and Ronen}{Nisan and Ronen}{2007}]%
        {NR07}
\bibfield{author}{\bibinfo{person}{Noam Nisan} {and} \bibinfo{person}{Amir
  Ronen}.} \bibinfo{year}{2007}\natexlab{}.
\newblock \showarticletitle{Computationally feasible VCG mechanisms}.
\newblock \bibinfo{journal}{\emph{Journal of Artificial Intelligence Research}}
   \bibinfo{volume}{29} (\bibinfo{year}{2007}), \bibinfo{pages}{19--47}.
\newblock


\bibitem[\protect\citeauthoryear{Rochet}{Rochet}{1987}]%
        {R87}
\bibfield{author}{\bibinfo{person}{Jean-Charles Rochet}.}
  \bibinfo{year}{1987}\natexlab{}.
\newblock \showarticletitle{A necessary and sufficient condition for
  rationalizability in a quasi-linear context}.
\newblock \bibinfo{journal}{\emph{Journal of mathematical Economics}}
  \bibinfo{volume}{16}, \bibinfo{number}{2} (\bibinfo{year}{1987}),
  \bibinfo{pages}{191--200}.
\newblock


\bibitem[\protect\citeauthoryear{Vickrey}{Vickrey}{1961}]%
        {V61}
\bibfield{author}{\bibinfo{person}{William Vickrey}.}
  \bibinfo{year}{1961}\natexlab{}.
\newblock \showarticletitle{Counterspeculation, auctions, and competitive
  sealed tenders}.
\newblock \bibinfo{journal}{\emph{The Journal of finance}}
  \bibinfo{volume}{16}, \bibinfo{number}{1} (\bibinfo{year}{1961}),
  \bibinfo{pages}{8--37}.
\newblock


\bibitem[\protect\citeauthoryear{Weed, Perchet, and Rigollet}{Weed
  et~al\mbox{.}}{2016}]%
        {WRP16}
\bibfield{author}{\bibinfo{person}{Jonathan Weed}, \bibinfo{person}{Vianney
  Perchet}, {and} \bibinfo{person}{Philippe Rigollet}.}
  \bibinfo{year}{2016}\natexlab{}.
\newblock \showarticletitle{Online learning in repeated auctions}. In
  \bibinfo{booktitle}{\emph{Conference on Learning Theory}}.
  \bibinfo{pages}{1562--1583}.
\newblock


\end{thebibliography}
\appendix
\section{Appendix}

\subsection{Proof of Theorem~\ref{thm:rgt-bound-multiple-bids-nonfixed-value}}\label{app:proof-thm2}
We  first denote some necessary notations in this section. Let $C_t(v,b) = \frac{48(m+1)(2U)^2 \ln t}{n\cdot \Delta^2(v, b)}$. Let event $\zeta_t$ be $$\zeta_t = \Big\{\forall b\in \bvec_t, n_{t-1}(v_t)\wedge n_{t-1}(b)\geq C_t(v_t,b)\Big\}$$ Let $(v, \bvec_t) \equiv \big\{(v, b_t^1),\cdots,(v, b_t^m)\big\}$, and $\Delta(v, \bvec_t) = \min_{j\in[m]}\Delta(v, b_t^j)=rgt(v^*, b^*) - \max_{j\in [m]}rgt(v, b_t^j)$. Moreover, let event $\zeta_{t, b}$ be $\big\{n_{t-1}(v_t)\wedge n_{t-1}(b)\geq C_t(v_t,b)\big\}$. Following the similar proof steps in Theorem~\ref{thm:rgt-bound-multiple-bids}, we first bound $\E[R(T)]$ as below,
\begin{equation}\label{eq:pseudo-rgt-2}
\begin{aligned}
&\E[R(T)]\\
& = \max_{v, b\in B}\left[\sum_{t=1}^T rgt(v, b)\right] - \E\left[\sum_{t=1}^T\widehat{rgt}_t(v_t, \bvec_t)\right]\\
& \overset{(a)}{=} \sum_{t=1}^T rgt(v^*, b^*) - \sum_{t=1}^T \E\left[\E\Big[\max_{j\in [m]} \tilde{g}^j_t(b_t^j)\cdot v_t - \tilde{p}^j_t(b_t^j)\Big| v_t, \bvec_t\Big]\right] + \sum_{t=1}^T \E\left[\E\Big[\tilde{g}^{m+1}(v_t)\cdot v_t - \tilde{p}_t^{m+1}(v_t)\Big] \Big| v_t\right]\\
& \overset{(b)}{\leq} \sum_{t=1}^T rgt(v^*, b^*) - \sum_{t=1}^T \E\left[\max_{j\in [m]} \E\Big[\tilde{g}^j_t(b_t^j)\cdot v_t - \tilde{p}^j_t(b_t^j)\Big| v_t, \bvec_t\Big]\right]+ \sum_{t=1}^T \E\left[u^*(v_t, v_t)\right]\\
& = \sum_{t=1}^T rgt(v^*, b^*) -\E\left[\max_{j\in [m]} rgt(v_t, b_t^j)\right]\\
& \leq \E\Bigg[\sum_{t=1}^T \1\big\{(v^*, b^*)\notin (v_t, \bvec_t)\big\}\cdot \big(rgt(v^*, b^*)-\max_{j\in [m]} rgt(v_t, b_t^j) \big)\Bigg]\\
& = \E\left[\sum_{t=1}^T\1\big\{(v^*, b^*) \notin (v_t, \bvec_t)\big\}\cdot \Delta(v_t, \bvec_t)\right]
\end{aligned}
\end{equation}
where inequality $(a)$ is based on the separation of expectation trick used in Equation~\ref{eq:ub-pseudo-regret1}$(b)$ and inequality $(b)$ holds because $\E[\max\{X_1, \cdots, X_n\}]\geq \max\{\E[X_1], \cdots, E[X_n]\}$. Then we separate the upper bound of $\E\left[R(T)\right]$ in Equations~(\ref{eq:pseudo-rgt-2}) into two terms,
\begin{equation*}
\begin{aligned}
&\E[R(T)] \leq \underbrace{\E\left[\sum_{t=1}^T\1\big\{(v^*, b^*) \notin (v_t, \bvec_t), \zeta_t\big\}\cdot \Delta(v_t, \bvec_t)\right]}_{\textbf{Term 1}} + \underbrace{\E\left[\sum_{t=1}^T\1\big\{(v^*, b^*) \notin (v_t, \bvec_t), \neg \zeta_t\big\}\cdot \Delta(v_t, \bvec_t)\right]}_{\textbf{Term 2}}
\end{aligned}
\end{equation*}

Finally we bound \textbf{Term 1} and \textbf{Term 2} respectively. For \textbf{Term 1}, we do the following decomposition,
\begin{align*}
&\textbf{Term 1} \leq \E\left[\sum_{t=1}^T\1\big\{(v^*, b^*) \notin (v_t, \bvec_t), \zeta_t\big\}\cdot \Delta(v_t, \bvec_t)\right]\\
& \leq \E\left[\sum_{t=1}^T\Delta(v_t, \bvec_t) \cdot\Big(\1\big\{v_t = v^*, b^*\notin \bvec_t, \zeta_t\big\} + \1\big\{v_t\neq v^*, \zeta_t\big\}\Big)\right]\\
&\leq \underbrace{\E\Bigg[\sum_{t=1}^T \1\big\{v_t=v^*, b^*\notin \bvec_t, \zeta_t\big\}\cdot \Delta(v_t, \bvec_t)\Bigg]}_{(*)} + \underbrace{\E\Bigg[\sum_{t=1}^T \1\big\{v_t\neq v^*, \zeta_t\big\}\cdot \Delta(v_t, \bvec_t)\Bigg]}_{(\#)}
\end{align*}
For term $(*)$, we have
\begin{align*}
(*) &\leq  \E\Bigg[\sum_{t=1}^T \1\big\{v_t = v^*, \mathrm{UCB}^\text{rgt}_t(v_t, b^j_t) \geq \mathrm{UCB}^\text{rgt}_t(v^*, b^*),\forall j\in [m], \zeta_t\big\}\cdot \Delta(v_t, \bvec_t)\Bigg]\\
&\leq \E\Bigg[\sum_{t=1}^T \frac{\Delta(v_t, \bvec_t)}{m}\cdot\sum_{j=1}^m \1\big\{\mathrm{UCB}^\text{u}_t(v^*, b_t^j)\geq\mathrm{UCB}^\text{u}_t(v^*, b^*), \zeta_t\big\}\Bigg]\\
&\overset{(c)}{\leq} \sum_{\substack{b\neq b^*}}\frac{\Delta(v^*, b)}{m}\cdot\left(1 + \sum_{t=1}^T\frac{4}{t^2}\right) \leq  \sum_{b\neq b^*}\frac{2\pi^2\Delta(v^*, b)}{3m}
\end{align*}
where the last inequality above holds based on Lemma~\ref{lem:rgt-utility}. Next we consider term $(\#)$,
\begin{align*}
&(\#)\leq \E\left[\sum_{t=1}^T \Delta(v_t, \bvec_t)\cdot\1\big\{v_t\neq v^*, \mathrm{UCB}^\text{rgt}_t(v_t, b_t^1)\geq \mathrm{UCB}^\text{rgt}_t(v^*, b^*), \zeta_t\big\}\right]\\
& \leq \sum_{v,b\in B, v\neq v^*} \Delta(v,b) \cdot \sum_{t=1}^T\PP\Big(\mathrm{UCB}^\text{rgt}_t(v_t, b)\geq \mathrm{UCB}^\text{rgt}_t(v^*, b^*)\Big|\zeta_{t,b}\Big)\\
& \overset{\text{Lemma}~\ref{lem:rgt-utility}}{\leq} \sum_{v,b\in B, v\neq v^*} \Delta(v,b) \cdot \Big(1 + \sum_{t=2}^T \frac{4}{t^2}\Big)\\
& = \sum_{v,b\in B, v\neq v^*} \frac{2\pi^2\Delta(v, b)}{3}
\end{align*}
Thus, $\textbf{Term 1}$ can be bounded by
\begin{equation}\label{ie:non-fix-value-term1}
\textbf{Term 1}\leq \sum_{\substack{v, b\in B, \\ (v,b)\neq(v^*, b^*)}} \frac{2\pi^2\Delta(v, b)}{3}\cdot\left(\1\big\{v\neq v^*\big\} + \frac{\1\big\{v= v^*\big\}}{m}\right)
\end{equation}
What remains is to bound \textbf{Term 2},  
\begin{equation}\label{ie:non-fix-value-term2}
\begin{aligned}
&\textbf{Term 2} \\
&\leq \E\Bigg[\sum_{t=1}^T \Delta(v_t, \bvec_t) \cdot \Big(\1\big\{\exists b\in \bvec_t, n_{t-1}(v_t) \leq C_t(v_t, b)\big\} + \1\big\{\exists b\in \bvec_t, n_{t-1}(b) \leq C_t(v_t, b)\big\}\Big)\Bigg]\\
& \leq \E\Bigg[\sum_{t=1}^T \sum_{b\in b_t}\Delta(v_t, b) \cdot \Big(\1\big\{n_{t-1}(v_t) \leq C_t(v_t, b) \big\}+ \1\big\{n_{t-1}(b) \leq C_t(v_t, b)\big\}\Big)\Bigg]\\
& \leq \sum_{\substack{v, b\in B, \\ (v,b)\neq(v^*, b^*)}} \Delta(v, b) \cdot\Big(\sum_{t=1}^T\1\big\{n_{t-1}(v) \leq C_t(v, b) \big\}+ \1\big\{n_{t-1}(b) \leq C_t(v, b)\big\}\\
& \leq \sum_{\substack{v, b\in B, \\ (v,b)\neq(v^*, b^*)}} 2\Delta(v, b) C_T(v, b) = \frac{96(m+1)(2U)^2\ln T}{n\Delta(v,b)}
\end{aligned}
\end{equation}

Combining inequalities~(\ref{ie:non-fix-value-term1}) and (\ref{ie:non-fix-value-term2}) we complete the proof.

\subsection{Proof of Theorem~\ref{thm:rgt-bound-switching}}\label{app:proof-thm-switching}
This proof is a straightforward extension of the proof in Theorem~\ref{thm:rgt-bound-multiple-bids} and Theorem~\ref{thm:rgt-bound-multiple-bids-nonfixed-value}. We only show the significant difference of this proof. Let's first denote some different notations used in this section. Let event $$\zeta_t \triangleq \left\{\forall i,j\in [m+1], n_{t-1}\big(b^i_t\big) \wedge n_{t-1}\big(b^j_t\big)\geq \frac{48(m+1)(2U)^2 \ln t}{n\Delta^2\big(b^i_t, b^j_t\big)}\right\}$$
and
$$\zeta_{t, b_1, b_2}\triangleq \left\{n_{t-1}\big(b_1\big) \wedge n_{t-1}\big(b_2\big)\geq \frac{48(m+1)(2U)^2 \ln t}{n\Delta^2\big(b_1, b_2\big)}\right\}$$
Denote $\bvec_t\otimes \bvec_t \triangleq \left\{\big(b^i_t, b^j_t\big)\right\}_{i,j\in [m+1], i\neq j}$ and $\Delta(\bvec_t) = rgt(v^*, b^*) -\max_{i,j\in [m+1], i\neq j} rgt\big(b^i_t, b^j_t\big)$. Follow the same argument adopted in Section~\ref{sec:proof-rgt-bound1}, we first decompose the {\em Pseudo-Regret} in the following way,
\begin{equation}\label{eq:rgt-bound-3}
\begin{aligned}
\E\left[R(T)\right]
&\leq \sum_{t=1}^T rgt(v^*, b^*) - \sum_{t=1}^T \E\left[\E\Big[\max_{i,j\in [m+1], i\neq j} \tilde{u}^j_t\big(b_t^i, b_t^j\big) - \tilde{u}^i_t\big(b_t^i, b_t^i\big)\Big|\bvec_t\Big]\right]\\
& \leq \sum_{t=1}^T rgt(v^*, b^*) - \sum_{t=1}^T \E\left[\max_{i,j\in [m+1], i\neq j}u^*\big(b_t^i, b_t^j\big) - u^*\big(b_t^i, b_t^i\big)\right]\\
&\leq \E\left[\sum_{t=1}^T \1\big\{(v^*, b^*)\notin \bvec_t\otimes \bvec_t\big\}\cdot \Delta(\bvec_t)\right]\\
& \leq \underbrace{\E\left[\sum_{t=1}^T \1\big\{(v^*, b^*)\notin \bvec_t\otimes \bvec_t, \zeta_t\big\}\cdot \Delta(\bvec_t)\right]}_{(*)} + \underbrace{\E\left[\sum_{t=1}^T \1\big\{(v^*, b^*)\notin b_t\otimes b_t, \neg\zeta_t\big\}\cdot \Delta(\bvec_t)\right]}_{(\#)}
\end{aligned}
\end{equation}

What remains is to bound term $(*)$ and term $(\#)$. To bound $(*)$, we use the exactly same technique in Equations~\ref{ie:fix-value-term1}, the difference is that we need to consider $\mathrm{UCB}^\text{rgt}$ terms and any bids combination $(b_1, b_2)$ in $\bvec_t\otimes \bvec_t$,
\begin{equation}
\begin{aligned}
(*)& \leq\E\Bigg[\sum_{t=1}^T\frac{1}{m(m+1)}\cdot \Big[\sum_{\substack{b_1, b_2\in \bvec_t\\b_1\neq b_2}}\Delta(b_1, b_2)\cdot\1\big\{\mathrm{UCB}^\text{rgt}_t(b_1,b_2)\geq \mathrm{UCB}^\text{rgt}_t(v^*,b^*), \zeta_t\big\}\Big]\Bigg]\\
& \leq \sum_{\substack{b_1, b_2\in B}} \frac{\Delta(b_1, b_2)}{(m+1)^2}\cdot\sum_{t=1}^T\PP\left(\mathrm{UCB}^\text{rgt}_t(b_1,b_2)\geq \mathrm{UCB}^\text{rgt}_t(v^*,b^*)\Big| \zeta_{t, b_1, b_2}\right)\\
& \leq \sum_{\substack{b_1, b_2\in B}} \frac{\Delta(b_1, b_2)}{(m+1)^2}\cdot \left(1+\sum_{t=1}^T\frac{4}{t^2}\right) \leq \sum_{\substack{b_1, b_2\in B}} \frac{2\pi^2\Delta(b_1, b_2)}{3(m+1)^2}
\end{aligned}
\end{equation}

To bound $(\#)$, we utilize the same argument in Equations~(\ref{ie:non-fix-value-term2}).
\begin{equation}\label{ie:switching-term2}
\begin{aligned}
(\#)
&\leq \E\Bigg[\sum_{t=1}^T \Delta(\bvec_t) \cdot \1\Big\{\exists b^i_t, b^j_t\in b_t, n_{t-1}(b_t^i) \leq \frac{48(m+1)(2U)^2 \ln t}{n\cdot \Delta^2(b_t^i, b_t^j)} \Big\}\Bigg]\\
&\leq \E\left[\sum_{t=1}^T \sum_{b_1, b_2\in b_t} \Delta(\bvec_t)\cdot \1\Big\{n_{t-1}(b_1) \leq \frac{48(m+1)(2U)^2 \ln t}{n\cdot \Delta^2(b_1, b_2)} \Big\}\right]\\
& \leq \sum_{b_1, b_2\in B} \Delta(b_1, b_2)\sum_{t=1}^T \1\Big\{n_{t-1}(b_1) \leq \frac{48(m+1)(2U)^2 \ln t}{n\cdot \Delta^2(b_1, b_2)} \Big\}\\
&\leq  \sum_{b_1, b_2\in B}\frac{192(m+1)U^2 \ln T}{n\cdot \Delta(b_1, b_2)} 
\end{aligned}
\end{equation}

\end{document}